\documentclass{article}
\usepackage[utf8]{inputenc}
\usepackage[T1]{fontenc}
\usepackage{amsmath}
\usepackage{amssymb}
\usepackage{bussproofs}
\usepackage{amsthm}
\usepackage{tikz}
\usepackage{float}
\usepackage{xcolor}
\usepackage{enumitem}
\usepackage{mathtools}
\usepackage{hyperref}
\usepackage{pgfplots}
\pgfplotsset{compat=1.18}
\newlist{induction}{description}{2}
\setlist[induction,1]{
  leftmargin=0pt,
  labelindent=\parindent,
  labelsep=0.5em,
  align=left
}
\setlist[induction,2]{
  leftmargin=\parindent,
  labelindent=2\parindent,
  labelsep=0.5em,
  align=left
}

\pgfplotsset{
    every axis/.append style={
        label style={font=\scriptsize},
        tick label style={font=\scriptsize},
        title style={font=\small}
    }
}

\newtheorem{theorem}{Theorem}
\newtheorem{corollary}{Corollary}
\newtheorem{lemma}{Lemma}

\theoremstyle{definition}

\newtheorem{example}{Example}

\newtheoremstyle{appendix}
  {6pt}
  {6pt}
  {\itshape}
  {}
  {\bfseries}
  {{\bfseries .}}
  {5pt plus 1pt minus 1pt}
  {\thmname{#1} \thmnumber{#2} \thmnote{\normalfont#3}}

\theoremstyle{appendix}

\newtheorem*{applem}{Lemma}

\title{An algorithmic construction of Nullstellensatz certificates for graph
automorphism groups}
\author{Ivan Drecun\\
\small Faculty of Mathematics, University of Belgrade\\
\small Studentski trg 16, 11000 Belgrade, Serbia\\
\small \texttt{idrecun@gmail.com}
}
\date{}

\begin{document}
\maketitle
\begin{abstract}
  
  Solvers for graph isomorphism and related problems are widely used in both
  theoretical and applied domains, but their complexity often requires users to
  trust their correctness without verification. We present a certifying
  algorithm for computing a generating set of the automorphism group of a given
  graph, meaning that it produces an independently verifiable certificate that
  the output list of automorphisms generates the entirety of the automorphism
  group. The problem of constructing a certificate is reduced to certifying
  that a system of multivariate polynomial equations is infeasible. An
  individualization-refinement algorithm produces a derivation within an
  algebraic proof system, representing a Nullstellensatz certificate of
  infeasibility. We provide an implementation of both the certifying algorithm
  and a certificate verifier, and demonstrate their viability on graphs with
  hundreds of vertices.


\end{abstract}

	\section{Introduction}

    An \emph{isomorphism} between two graphs is a bijection between their sets
    of vertices such that the adjacency relation is preserved. An
    \emph{automorphism} is an isomorphism from a graph $G$ to itself, with the
    set of all automorphisms forming a group under composition, called the
    \emph{automorphism group} of $G$ and denoted by $\operatorname{Aut}(G)$.
    The \emph{graph isomorphism problem} (\textsc{GI}) is the decision problem
    of determining whether two given graphs are isomorphic. The graph
    isomorphism problem has been shown to be polynomial time equivalent to
    several closely related problems. One such problem is determining a
    generating set of the graph automorphism group \cite{Mathon}. Many systems
    for solving \textsc{GI} and related problems have been developed over the
    years \cite{McKayPiperno, bliss, saucy, conauto} and it may be argued that,
    from a practical standpoint, these problems have been solved. It might come
    as a surprise, then, that the theoretical bounds for the complexity of such
    problems are still not fully understood. While it is clear that \textsc{GI}
    is in \textsc{NP} (since an isomorphism is a polynomial-size certificate),
    the relation of \textsc{GI} to either \textsc{P} or \textsc{NP}-complete is
    not known. The best known upper bound on the computational complexity of
    \textsc{GI} is quasipolynomial \cite{Babai2016}. Furthermore, \textsc{GI}
    is not known to be in \textsc{coNP}, meaning that there is no known
    polynomial-size certificate for non-isomorphism of two graphs
    \cite{Babai2018}.

    A certifying algorithm is an algorithm which, along with its output,
    produces a certificate of its correctness. The certificate should be easily
    and independently verifiable. The benefit of certifying algorithms is that
    trust may be offloaded from the potentially complex algorithm and its
    implementation to a much simpler certificate verification procedure
    \cite{Certifying}. Such algorithms are particularly useful in automated and
    interactive theorem proving, allowing for complex algorithms to be used as
    primitives within proof assistants, without requiring the underlying
    implementation itself to be verified. A practical certifying algorithm for
    graph isomorphism, and more generally for graph canonical labeling, has
    been developed in recent years \cite{LMCS}. The produced certificate and
    its verification procedure are coupled to a specific instance of McKay and
    Piperno's graph canonical labeling algorithm, and the certificate is
    represented as a proof within a relatively complex, purpose-built proof
    system.

    In this paper, we present a certifying algorithm for calculating a set of
    generators of $\operatorname{Aut}(G)$ for some graph $G$. Specifically, we
    deal with the problem of constructing a certificate that a produced set
    $a_1, \dots, a_m$ of automorphisms of $G$ generates the entirety of
    $\operatorname{Aut}(G)$. The problem is reduced to deciding the
    infeasibility of a system of multivariate polynomial equations. Namely, the
    reduction is via the typical identification of the automorphisms with
    permutation matrices $X$ satisfying $AX=XA$, where $A$ is the adjacency
    matrix of $G$ \cite{Sachs1983, DeLoera2010, BerkholzGrohe2015}, with an
    additional polynomial family for excluding all automorphisms generated by
    $a_1, \dots, a_m$. If such a system is infeasible, then it must be that all
    automorphisms are generated by $a_1, \dots, a_m$. We introduce a simple
    algebraic proof system for certifying ideal membership whose correctness we
    prove, and within which we construct a Nullstellensatz certificate of
    infeasibility \cite{Nulla}. The construction is guided by an
    individualization-refinement search procedure similar to the one described
    by McKay and Piperno. More specifically, we associate a monomial with each
    node of the search tree (after pruning), with $1$ corresponding to its
    root. We give an explicit bottom-up derivation of those monomials within
    the proof system. The verification procedure checks that the supplied
    permutations are indeed automorphisms of the input graph, and that each
    step of the derivation is a valid application of a rule of the proof
    system, with the last step producing $1$, corresponding to the
    Nullstellensatz certificate. This requires only simple polynomial
    operations, i.e. addition, multiplication and substitution. We provide an
    implementation of the certifying algorithm, as well as the certificate
    verifier, and report on the results of the experimental evaluation.

	\section{An ideal membership proof system}
  \label{sec:proof}

  In this section we introduce a simple proof system based on the
  Nullstellensatz for verifying membership of a polynomial within a
  multivariate polynomial ideal. We then introduce a polynomial encoding of a
  graph's automorphism group such that a certificate for a set of generators of
  the group may be constructed within the proof system. Some basic definitions
  and known results follow.

	\subsection{Background}

  Let $K[x_1, \dots, x_k]$ denote the \emph{multivariate polynomial ring} with
  $k$ variables over a field $K$. An ideal $I \subseteq K[x_1, \dots, x_k]$ is
  \emph{generated} by $F = \{f_1, \dots, f_n \}$, a finite subset of $K[x_1,
  \dots, x_k]$, if $I = \langle F \rangle = \{ p_1 f_1 + \dots + p_n f_n \mid
  p_1, \dots, p_n \in K[x_1, \dots, x_k] \}$.

  An \emph{affine variety} of a set $F \subseteq K[x_1, \dots, x_k]$
  is the set of all points in $K^k$ which are the zeros of all $f \in
  F$ simultaneously. More formally, the affine variety of $F$ is the
  set $V(F) = \{(x_1, \dots, x_k) \in K^k \mid f(x_1, \dots, x_k) = 0
  \text{ for all } f \in F\}$. If $V(F) = \emptyset$, we say that $F$
  is \emph{infeasible}.

  The following theorem establishes the necessary and sufficient
  condition for an ideal of $K[x_1, \dots, x_k]$ to be infeasible. It
  is the weak form of Hilbert's Nullstellensatz \cite{Hilbert}.

	\begin{theorem}[Weak Nullstellensatz]
		Let $K$ be an algebraically closed field and let $I$ be an ideal of
		$K[x_1, \dots, x_k]$. Then, $V(I) = \emptyset$ if and only if $I = K[x_1,
		\dots, x_k]$.
	\end{theorem}

  In particular, if $I$ is generated by $F = \{f_1, \dots,
  f_n\}$, from the Nullstellensatz it follows that $f_1, \dots, f_n$
  have no common zeros (i.e. $F$ is infeasible) if and only if there
  exist polynomials $g_1, \dots, g_n \in K[x_1, \dots, x_k]$ such that
  $1 = \sum_{i=1}^n g_i f_i$. The sequence of such polynomials $g_1,
  \dots, g_n$, if they exist, is called a \emph{Nullstellensatz
  certificate}.

  We define any mapping $s : \{x_1, \dots, x_k\} \to \{x_1, \dots, x_k\}$ to be
  a \emph{substitution}. A substitution $s$ acts on $p \in K[x_1, \dots, x_k]$
  such that $p^s(x_1, \dots, x_k) = p(s(x_1), \dots, s(x_k))$. This action may
  be extended to sets of polynomials such that $F^s = \{p^s \mid p \in F\}$. We
  say that $F \subseteq K[x_1, \dots, x_k]$ is closed under a substitution $s$
  if $F^s \subseteq F$.

	\begin{example}
    Let $s: \begin{pmatrix}x & y \\ y & x \end{pmatrix}$ and $s' :
      \begin{pmatrix}x & y \\ y & y \end{pmatrix}$ be substitutions. Let $F_1 =
        \{x^2, xy, y^2\}$ and $F_2 = \{x, y, y^2\}$ be sets of polynomials.
        Then $F_1^s = \{y^2, yx, x^2\}$ and $F_2^s = \{y, x, x^2\}$, as well as
        $F_1^{s'} = \{y^2\}$ and $F_2^{s'} = \{y, y^2\}$. Clearly, $F_1$ is
        closed under $s$ and $s'$, but $F_2$ is closed only under $s'$ and not
        $s$.
  \end{example}

	\subsection{Proof system}

  We introduce a proof system for certifying the membership of a polynomial in
  a multivariate polynomial ideal. The proof system sits close to the members
  of the broader family of algebraic proof systems, such as the Nullstellensatz
  proof system \cite{Beame1996}, polynomial calculus
  \cite{CleggEdmondsImpagliazzo1996} and the ideal proof system
  \cite{GrochowPitassi2018}.

  Let $K$ be a field and $F = \{f_1, \dots, f_n\} \subseteq K[x_1, \dots, x_k]$
  a set of multivariate polynomials over $K$. The proof system consists of
  three rules for deriving polynomials contained in the polynomial ideal $I =
  \langle F \rangle \subseteq K[x_1, \dots, x_k]$.

	\begin{prooftree}
	\AxiomC{}
	\RightLabel{Tautology}
	\UnaryInfC{$f_i$}
	\end{prooftree}

	\begin{prooftree}
	\AxiomC{$p_1$}
	\AxiomC{$\dots$}
	\AxiomC{$p_m$}
    \RightLabel{Ideal [where $g_1, \dots, g_m \in K[x_1, \dots, x_k]$]}
	\TrinaryInfC{$\sum_{i=1}^m g_i p_i$}
	\end{prooftree}
	
	\begin{prooftree}
	\AxiomC{$p$}
    \RightLabel{Substitution [where $F^s \subseteq F$]}
	\UnaryInfC{$p^s$}
	\end{prooftree}

	\begin{example}
    Let $F = \{ x_1 x_2, x_1 + x_2 - 1 \}$ over $\mathbb{R}$. Note that $F$ is
    closed under $s : \begin{pmatrix} x_1 & x_2 \\ x_2 & x_1
    \end{pmatrix}$ since $(x_1 x_2)^s = x_2 x_1 = x_1 x_2$ and $(x_1 + x_2
    - 1)^s = x_2 + x_1 - 1 = x_1 + x_2 - 1$. A proof of $x_1^2 + x_2^2 - 1 \in
    \langle F \rangle$ might look like the following.
		\begin{enumerate}
			\item $x_1 x_2$ (Tautology)
			\item $x_1 + x_2 - 1$ (Tautology)
			\item $x_1^2 - x_1$ (Ideal, $x_1 \cdot (x_1 + x_2 - 1) - x_1 x_2$)
			\item $x_2^2 - x_2$ (Substitution, $(x_1^2 - x_1)^s$)
			\item $x_1^2 + x_2^2 - 1$ (Ideal, $(x_1^2 - x_1) + (x_2^2 - x_2) + (x_1 + x_2 - 1)$)
		\end{enumerate}
	\end{example}

  Note that the inclusion of the Substitution rule is not strictly necessary
  for the completeness of the proof system. However, it is important since the
  rule can be used to avoid producing duplicate sequences of steps of mutually
  symmetric parts of the proof. This is similar to the symmetry rule Krishnamurthy
  introduced for resolution \cite{Krishnamurthy1985}.

  The following two theorems establish the correctness of the proof system.
  Since in a later section we provide an explicit construction of a proof, we
  are more concerned with its soundness, rather than its completeness.

	\begin{theorem}[Soundness]
		\label{thm:soundness}
		If $q$ is derivable from $F$, then $q \in I$.
	\end{theorem}
	\begin{proof}
		We prove this by induction on the structure of the proof.

		\begin{itemize}
			\item[\textbf{T}] Since $f_i \in F$, then $f_i \in I$.
			\item[\textbf{I}] By the induction hypothesis it holds that $p_i = \sum_{j=1}^n g_{i, j}
				f_j$. Then $\sum_{i=1}^m g_i p_i = \sum_{i=1}^m g_i \sum_{j=1}^n g_{i, j} f_j
				= \sum_{j=1}^n (\sum_{i=1}^m g_i g_{i, j}) f_j \in I$.
			\item[\textbf{S}] By the induction hypothesis it holds that $p =
				\sum_{i=1}^n g_i f_i$.  Therefore $p^s = (\sum_{i=1}^n g_i f_i)^s =
				\sum_{i=1}^n g_i^s f_i^s \in \langle F^s \rangle \subseteq \langle F \rangle = I$.
		\end{itemize}
	\end{proof}

	\begin{corollary}
    \label{cor:soundness}
		If $1$ is derivable from $F$, then $V(I) = \emptyset$.
	\end{corollary}
	\begin{proof}
    If $1$ is derivable from $F$, by Theorem \ref{thm:soundness} it holds
    that $1 \in I$, meaning that there are no points at which the polynomials
    of $I$ are all simultaneously zero.
	\end{proof}	

	\begin{theorem}[Completeness]
		If $q \in I$, then $q$ is derivable from $F$.
	\end{theorem}
	\begin{proof}
		Since $q \in I$, there exist some $g_1, \dots, g_n$ such that $q =
		\sum_{i=1}^n g_i f_i$ which is clearly derivable by a combination of
		applications of the Tautology rule and a single application of the Ideal rule.
	\end{proof}

	\begin{corollary}
		If $K$ is an algebraically closed field and $V(I) = \emptyset$, then $1$ is derivable from $F$.
	\end{corollary}
	\begin{proof}
		Since $K$ is algebraically closed and $V(I) = \emptyset$, there exists a
		Nullstellensatz certificate $1 = \sum_{i=1}^n g_i f_i$. Then $1$ is derivable
		due to the completeness of the proof system.
	\end{proof}

  A derivation of $1$ within the proof system shows that a Nullstellensatz
  certificate for the given set of polynomials exists. An explicit
  Nullstellensatz certificate can be reconstructed by flattening the derivation
  tree and applying the additions, multiplications and substitutions
  corresponding to the rules within the tree.

	\subsection{A polynomial encoding of $\operatorname{Aut}(G)$}

  We now present a reduction of the problem of deciding whether a given set of
  automorphisms of a graph generates the entire automorphism group, i.e.
  $\operatorname{Aut}(G) = \langle a_1, \dots, a_m \rangle$, to the decision
  problem of the infeasibility of a system of multivariate polynomial
  equations. We introduce a set of polynomials $\mathcal{A}(G, a_1, \dots,
  a_m)$ such that it has no zeros (i.e. the affine variety $V(\mathcal{A}(G,
  a_1, \dots, a_m))$ is empty) if and only if the equality
  $\operatorname{Aut}(G) = \langle a_1, \dots, a_m \rangle$ holds.

  Let $G$ be an undirected graph on $n$ vertices with adjacency matrix $A$, and
  let $a_1, \dots, a_m$ be some automorphisms of $G$. Let $X = (x_{i, j})_{0
  \leq i, j < n}$ be a matrix of indeterminates over the finite field $\mathbb{F}_p$,
  where $p > n$ is prime. We define $\mathcal{A}(G, a_1, \dots, a_m)$ as the
  union of the following families of polynomials over $\mathbb{F}_p$:
	\begin{itemize}
    \item $\operatorname{Bin}_{i, j} = x_{i, j}^2 - x_{i, j}$ for $0 \leq i, j
      < n$. The family $\operatorname{Bin}$ ensures that each element of $X$ is
      either 0 or 1.
    \item $\operatorname{RowNZ}_{i; j, l} = x_{i, j} x_{i, l}$ for $0
      \leq i, j, l < n$ and $j \neq l$. The family $\operatorname{RowNZ}$
      ensures that no row of $X$ has more than one non-zero element.
    \item $\operatorname{ColNZ}_{i, l; j} = x_{i, j} x_{l, j}$ for $0
      \leq i, j, l < n$ and $i \neq l$. The family $\operatorname{ColNZ}$
      ensures that no column of $X$ has more than one non-zero element.
    \item $\operatorname{RowS}_i = \sum_{j=0}^{n-1} x_{i, j} - 1$ for $0 \leq i < n$. The family
      $\operatorname{RowS}$ ensures that each row of $X$ sums to 1.
    \item $\operatorname{ColS}_j = \sum_{i=0}^{n-1} x_{i, j} - 1$ for $0 \leq j < n$. The family
      $\operatorname{ColS}$ ensures that each column of $X$ sums to 1.
    \item $\operatorname{Iso}_{i, j} = (AX - XA)_{i, j}$ for $0 \leq i, j < n$. The family
      $\operatorname{Iso}$ ensures that $X$ corresponds to an automorphism of
      $G$.
    \item $\operatorname{ExIso}_a = \prod_{u = 0}^{n-1} x_{u, a(u)}$ for $a \in \langle a_1, \dots, a_m \rangle$. The family
      $\operatorname{ExIso}$ ensures that $X$ does not correspond to any
      automorphism in $\langle a_1, \dots, a_m \rangle$.
	\end{itemize}

	\begin{example}
    Let $G$ be the graph given in Figure \ref{img:graphs}. The automorphisms of
    $G$ are $a_0 = id$ and $a_1 = (0\ 2)$. We may construct the following set
    of polynomials.

		\begin{figure}[H]
		\begin{center}
			\begin{tikzpicture}
  \begin{scope}
    \node (a1) at (0, 0)[draw, circle, label=180:$0$] {};
    \node (a2) at (0.5, 0.86)[draw, circle, label=90:$1$] {};
    \node (a3) at (1, 0)[draw, circle, label=0:$2$] {};

		\node (am) at (3, 0.5)[] {$\begin{bmatrix}0 & 1 & 0\\1 & 0 & 1\\0 & 1 & 0\end{bmatrix}$};
  \end{scope}
  \begin{scope}
    \path [-] (a1) edge (a2);
    \path [-] (a2) edge (a3);
  \end{scope}
\end{tikzpicture}
		\end{center}
		\caption{An example graph with its adjacency matrix.}
		\label{img:graphs}
		\end{figure}
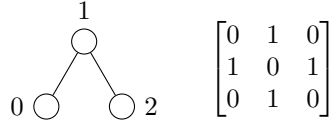
		\begin{equation*}
      \begin{split}
        \operatorname{Bin} = \{ & x_{0, 0}^2 - x_{0, 0},\ x_{0, 1}^2 - x_{0, 1},\ x_{0, 2}^2 - x_{0, 2}, \\
              & x_{1, 0}^2 - x_{1, 0},\ x_{1, 1}^2 - x_{1, 1},\ x_{1, 2}^2 - x_{1, 2}, \\
              & x_{2, 0}^2 - x_{2, 0},\ x_{2, 1}^2 - x_{2, 1},\ x_{2, 2}^2 - x_{2, 2} \}\\
        \operatorname{RowNZ} =  \{ & x_{0,0}x_{0,1},\ x_{0,0}x_{0,2},\ x_{0,1}x_{0,2}, \\
              & x_{1,0}x_{1,1},\ x_{1,0}x_{1,2},\ x_{1,1}x_{1,2},\\
              & x_{2,0}x_{2,1},\ x_{2,0}x_{2,2},\ x_{2,1}x_{2,2} \}\\
        \operatorname{ColNZ} = \{ & x_{0,0}x_{1,0},\ x_{0,0}x_{2,0},\ x_{1,0}x_{2,0},\\
              & x_{0,1}x_{1,1},\ x_{0,1}x_{2,1},\ x_{1,1}x_{2,1},\\
              & x_{0,2}x_{1,2},\ x_{0,2}x_{2,2},\ x_{1,2}x_{2,2} \} \\
        \operatorname{RowS} = \{ & x_{0,0}+x_{0,1}+x_{0,2}-1,\ x_{1,0}+x_{1,1}+x_{1,2}-1,\ x_{2,0}+x_{2,1}+x_{2,2}-1 \} \\
        \operatorname{ColS} = \{ & x_{0,0}+x_{1,0}+x_{2,0}-1,\ x_{0,1}+x_{1,1}+x_{2,1}-1,\ x_{0,2}+x_{1,2}+x_{2,2}-1 \} \\
        \operatorname{Iso} = \{ & x_{1,0} - x_{0,1},\ x_{1, 1} - x_{0, 0} - x_{0, 2},\ x_{1, 2} - x_{0, 1},\\
          & x_{0, 0} + x_{2, 0} - x_{1, 1},\ x_{0, 1} + x_{2, 1} - x_{1, 0} - x_{1, 2},\ x_{0, 2} + x_{2, 2} - x_{1, 1},\\
          & x_{1,0} - x_{2,1},\ x_{1, 1} - x_{2, 0} - x_{2, 2},\ x_{1, 2} - x_{2, 1} \} \\
        \operatorname{ExIso} = \{ & x_{0,0}x_{1,1}x_{2,2},\ x_{0,2}x_{1,1}x_{2,0} \}
      \end{split}
		\end{equation*}
	\end{example}

	\begin{theorem}
    \label{thm:system_generators}
    $\mathcal{A}(G, a_1, \dots, a_m)$ is infeasible if and only if
    $\operatorname{Aut}(G) = \langle a_1, \dots, a_m \rangle$.
	\end{theorem}
	\begin{proof}
    We prove the contraposition of both implications, i.e. that an $X \in
    V(\mathcal{A}(G, a_1, \dots, a_m))$ exists if and only if there is an
    automorphism $a \in \operatorname{Aut}(G) \setminus \langle a_1, \dots, a_m
    \rangle$.
		\begin{itemize}
      \item[($\Rightarrow$)] Assume $X$ is a zero of $\mathcal{A}(G, a_1,
        \dots, a_m)$. Since the polynomials in $\operatorname{RowNZ}$ are zero,
        each row of $X$ must have at most one non-zero element. Similarly, each
        column of $X$ must have at most one non-zero element. Since the
        polynomials in $\operatorname{RowS}$ are zero, each row of $X$ must sum
        up to 1. This implies that exactly one element of the row is 1, and the
        rest must be 0. Similarly, each column of $X$ must have exactly one
        element equal to 1, and the rest must be 0. This means that $X$ is a
        permutation matrix.

        Since the polynomials in $\operatorname{Iso}$ are zero, and since
        $\operatorname{Iso}_{i, j}=(AX-XA)_{i, j}$, this means that $AX - XA = 0$.
        A well-known implication of this equality is that $X$ corresponds to an
        automorphism of $G$ \cite{Sachs1983}.

        Finally, since the polynomials in $\operatorname{ExIso}$ are zero, $X$
        cannot be the permutation matrix of any automorphism in $\langle a_1,
        \dots, a_m \rangle$. Let $a$ be one such automorphism and let
        $Y=(y_{i,j})_{0 \leq i,j < n}$ be its corresponding permutation matrix,
        i.e. $y_{i,a(i)} = 1$ for $0 \leq i < n$, and $y_{i, j} = 0$ otherwise.
        Since $\operatorname{ExIso}_a$ is zero, it means that at least one of
        $x_{u, a(u)}$ is zero, meaning $X \neq Y$.

      \item[($\Leftarrow$)] Assume there is an automorphism $a \in
        \operatorname{Aut}(G) \setminus \langle a_1, \dots, a_m \rangle$ and
        let $X$ be the permutation matrix corresponding to $a$. As $X$ is a
        permutation matrix, it is a zero of the polynomials of $\operatorname{Bin}$,
        $\operatorname{RowNZ}$, $\operatorname{ColNZ}$, $\operatorname{RowS}$
        and $\operatorname{ColS}$. As $a$ is an automorphism of $G$, $X$ is a
        zero of the polynomials of $\operatorname{Iso}$. Finally, as $a \notin
        \langle a_1, \dots, a_m \rangle$, $X$ is a zero of the polynomials of
        $\operatorname{ExIso}$. Therefore, $X \in V(\mathcal{A}(G, a_1, \dots,
        a_m))$.
		\end{itemize}
	\end{proof}

  The choice of field $\mathbb{F}_p$ where $p > n$ is such that the
  coefficients of the polynomials derived from $\mathcal{A}(G, a_1, \dots,
  a_m)$ may carry the cardinalities of subsets of vertices of $G$. We rely on
  this fact when constructing a certificate of infeasibility in Section
  \ref{sec:construction}. Note that Theorem \ref{thm:system_generators} does
  not depend on this choice of $p$.

  Note that $\mathcal{A}(G, a_1, \dots, a_m)$ contains redundant polynomials as
  defined. For example, excluding the $\operatorname{Bin}$ family or one of the
  $\operatorname{RowNZ}$ and $\operatorname{ColNZ}$ families does not change
  the zeros of the system. We include them regardless in order to simplify the
  construction of certificates.

  Another thing to take note of is the cardinality of $\operatorname{ExIso}$.
  Assuming $a_1, \dots, a_m$ generate the entirety of $\operatorname{Aut}(G)$,
  it would mean that $|\operatorname{ExIso}|=|\operatorname{Aut}(G)|$ which
  could be extremely large. In practice, however, this will not be an issue,
  since there will never be a need to instantiate all polynomials in
  $\mathcal{A}(G, a_1, \dots, a_m)$, but only the ones needed to complete a
  proof within the proof system.

	\begin{theorem}
    \label{thm:closed}
    Let $a \in \langle a_1, \dots, a_m \rangle$. Then $\mathcal{A}(G, a_1,
    \dots, a_m)$ is closed under the substitution $s$ given by $s(x_{i, j}) =
    x_{a(i), j}$.
	\end{theorem}
	\begin{proof}
    We show that each type of polynomial in $\mathcal{A}(G, a_1, \dots, a_m)$ is
    closed under the action of $s$.
		\begin{itemize}
      \item $\operatorname{Bin}_{i,j}^s = (x_{i,j}^2 - x_{i,j})^s =
        x_{a(i),j}^2 - x_{a(i),j} = \operatorname{Bin}_{a(i), j} \in
        \operatorname{Bin}$
      \item $\operatorname{RowNZ}_{i;j,l}^s = (x_{i,j} x_{i,l})^s = x_{a(i), j}
        x_{a(i), l} = \operatorname{RowNZ}_{a(i); j,l} \in
        \operatorname{RowNZ}$
      \item $\operatorname{ColNZ}_{i,l;j}^s = (x_{i,j} x_{l,j})^s = x_{a(i), j}
        x_{a(l), j} = \operatorname{ColNZ}_{a(i), a(l); j} \in
        \operatorname{ColNZ}$
      \item $\operatorname{RowS}_i^s = (\sum_{j=0}^{n-1} x_{i, j} - 1)^s =
        \sum_{j=0}^{n-1} x_{a(i), j} - 1 = \operatorname{RowS}_{a(i)} \in
        \operatorname{RowS}$
      \item $\operatorname{ColS}_j^s = (\sum_{i=0}^{n-1} x_{i, j} - 1)^s =
        \sum_{i=0}^{n-1} x_{a(i), j} - 1 = \operatorname{ColS}_{j} \in
        \operatorname{ColS}$
      \item Let $P_a$ be the permutation matrix of $a$. Note that $X^s = P_aX$,
        as well as $AP_a = P_aA$ since $a$ is an automorphism of $G$. It holds
        that
        \begin{equation*}
          \begin{split}
            (AX-XA)^s & = (AX^s-X^sA) \\
            & = A(P_aX) - (P_aX)A \\
            & = P_aAX - P_aXA \\
            & = P_a(AX-XA)
          \end{split}
        \end{equation*}
        meaning 
        \begin{equation*}
          \begin{split}
            \operatorname{Iso}_{i, j}^s & = (AX-XA)_{i, j}^s \\
            & = (P_a(AX-XA))_{i, j} \\
            & = (AX-XA)_{a(i), j} \\
            & = \operatorname{Iso}_{a(i), j} \in \operatorname{Iso}
          \end{split}
        \end{equation*}
      \item Let $b \in \langle a_1, \dots, a_m \rangle$. Then
        $\operatorname{ExIso}_b^s = (\prod_{u=0}^{n-1} x_{u, b(u)})^s =
        \prod_{u=0}^{n-1} x_{a(u), b(u)}$. Since both $a$ and $b$ are generated
        by $\{a_1, \dots, a_m\}$, that is true of $c = b \circ a^{-1}$ as well.
        Continuing the sequence of equalities, we have $\dots =
        \prod_{u=0}^{n-1} x_{a(u), c(a(u))} = \prod_{t=0}^{n-1} x_{t, c(t)} =
        \operatorname{ExIso}_{c} \in \operatorname{ExIso}$ for $t = a(u)$.
		\end{itemize}
	\end{proof}

  For the rest of the paper we will use the shorthand
  notation $p^a = p^s$ for $s$ defined as in the
  statement of Theorem \ref{thm:closed}.

  \section{Constructing the certificate of infeasibility}
  \label{sec:construction}

  In this section we describe the algorithm for determining a generating set
  of $\operatorname{Aut}(G)$ given a graph $G$. We then describe the
  construction of a proof of the infeasibility of $\mathcal{A}(G, a_1, \dots,
  a_m)$ based on this algorithm.

	\subsection{Notation and definitions}

  A sequence of elements $v_1, \dots, v_k$ of set $V$ is denoted by $[v_1,
  \dots, v_k]$. The prefix of length $k$ of a sequence
  $\nu$ is denoted by $[\nu]_k$. The sequence obtained by appending $v \in
  V$ to the end of sequence $\nu$ is written as $[\nu, v]$.

  A \emph{partition} $P$ of set $V$ is a set of non-empty subsets of $V$
  (called \emph{cells}) such that any two cells are disjoint and the union of
  all cells equals $V$. Let $[v]_P$ denote the cell of partition $P$ containing
  $v$. We say that a partition
  $P$ is \emph{discrete} if every cell of $P$ contains exactly one element.

  Let $P$ be a partition of a finite set $V$. A \emph{partition mapping} is any
  mapping $\mu : P \to \mathcal{P}(V)$ such that the sets in $\mu(P)$ are
  pairwise disjoint. We may use the shorthand notation $\mu[u] = \mu([u]_P)$,
  and $\mu^{-1}[v]$ for $C \in P$ such that $v \in \mu(C)$, or $\emptyset$ if no
  such cell $C$ exists. A partition mapping $\mu$ is \emph{bijective} if $|C| =
  |\mu(C)|$ for every $C \in P$. Otherwise, the partition mapping is
  \emph{partial}. Note that $\mu$ being bijective implies that $\mu(P)$ is a
  partition of $V$, while the reverse implication need not hold. We define
  $\sigma_\mu$ to be the mapping $u \overset{\sigma_\mu}{\mapsto} \min \mu[u]$
  when $\mu$ is bijective. Note that if $P$ is discrete, then $\sigma_\mu$ is a
  permutation of $V$.
  
  \begin{example}
    Let $P = \{\{0, 2\}, \{1, 3, 4\}\}$ be a partition of set $V = \{0, 1, 2,
    3, 4\}$. Then $\mu = \begin{smallmatrix}\{0, 2\} \{1, 3, 4\} \\ \{0, 3\}
    \{1, 2, 4\}\end{smallmatrix}$ is a bijective partition mapping since $|\{0,
    2\}| = |\{0, 3\}|$ and $|\{1, 3, 4\}| = |\{1, 2, 4\}|$. We have $\mu[0] =
    \mu[2] = \{0, 3\}$ and $\mu[1] = \mu[3] = \mu[4] = \{1, 2, 4\}$, as well as
    $\mu^{-1}[0] = \mu^{-1}[3] = \{0, 2\}$ and $\mu^{-1}[1] = \mu^{-1}[2] =
    \mu^{-1}[4] = \{1, 3, 4\}$, and finally $\sigma_{\mu} : \begin{pmatrix} 0
    & 1 & 2 & 3 & 4 \\ 0 & 1 & 0 & 1 & 1 \end{pmatrix}$.
  \end{example}

  \begin{example}
    Let $P = \{\{0, 2\}, \{1, 3, 4\}\}$ be a partition of set $V = \{0, 1, 2,
    3, 4\}$. Then both $\mu_1 = \begin{smallmatrix} \{0, 2\} \{1, 3, 4\} \\
      \{\} \{1, 2\} \end{smallmatrix}$ and $\mu_2 = \begin{smallmatrix} \{0,
    2\} \{1, 3, 4\} \\ \{0, 1, 3\} \{2, 4\} \end{smallmatrix}$ are partial
    partition mappings.
  \end{example}

  The number of vertices belonging to set $W$ adjacent to $u$ in graph $G$ is
  the \emph{degree of $u$ with respect to $W$}, denoted by
  $\operatorname{deg}(u, W)$. The
  set of neighbors of vertex $u$ in graph $G$ is denoted by $n_G(u)$. We may
  define a family of equivalence relations on the vertices of $G$ by $u \sim_W
  v$ if and only if $\operatorname{deg}(u, W) = \operatorname{deg}(v, W)$.

	\subsection{The algorithm for automorphism discovery}

  In this section we will present the details of the algorithm for discovering
  a set of generators of the automorphism group of a given graph. The algorithm
  is inspired by the algorithm described by McKay and Piperno \cite{McKayPiperno}.

  The algorithm constructs a search tree $\mathcal{T}(G)$ of graph $G$. The nodes
  of the search tree are represented by sequences $\nu = [u_1, \dots, u_k]$ of
  distinct vertices of $G$, with the root of the tree being represented by the
  empty sequence $[]$. The subtree of the search tree with root $\nu$ is
  denoted by $\mathcal{T}(G, \nu)$.

  Each node of the tree is assigned a partition mapping $P(G, \nu)$. The
  partition mapping is bijective for all inner nodes of the tree, and it may
  only be partial at its leaves. Let $T(G, \nu)$ be a function returning any
  non-singleton cell of the domain partition of $P(G, \nu)$ if $P(G, \nu)$ is
  bijective, and $\emptyset$ if it is partial or if no such cell exists. The
  search tree is defined such that children of node $\nu$ are nodes $[\nu,
  u_1], \dots, [\nu, u_k]$ for each vertex $u_i \in T(G, \nu)$. In order to
  define the partition mapping $P(G, \nu)$, we first define two operations,
  \emph{individualization} and \emph{refinement}.

  Let $\mu : \mathbb{U} \to \mathbb{V}$ be a bijective partition mapping and
  let $u \in U \in \mathbb{U}$ and $v = \sigma_\mu(u)$. We define the
  individualization operation as $i(\mu, u) = \mu'$ such that $\mu'(\{u\}) =
  \{v\}$, $\mu'(U \setminus \{u\}) = \mu(U) \setminus \{v\}$ and $\mu'(C) =
  \mu(C)$ otherwise.

  \begin{example}
  Let $\mu = \begin{smallmatrix}\{0, 2\} \{1, 3, 4\} \\ \{0, 3\} \{1, 2, 4\}
    \end{smallmatrix}$ be a partition mapping. Then $i(\mu, 3) =
    \begin{smallmatrix} \{0, 2\} \{3\} \{1, 4\} \\ \{0, 3\} \{1\} \{2, 4\}
    \end{smallmatrix}$ since $\mu[3] = \{1, 2, 4\}$ and
    $\min\mu[3] = 1$.
  \end{example}

  Let $U$ be the domain partition of a partition mapping $\mu$, and let $U'$ be
  the partition obtained by replacing each cell $C$ of $U$ with the cells $C /
  {\sim}_W$ for some cell $W$ of $U$. Define the \emph{split} operation $s(G,
  \mu, W)$ as the partition mapping $\mu'$ with domain partition $U'$. Let $C
  \in U$ and $C' \in C / {\sim}_W$. Define $\mu'$ such that $\mu'(C') = \{v \in
  \mu(C) \mid \operatorname{deg}(u, W) = \operatorname{deg}(v, \mu(W)) \text{
    for all } u \in C'\}$.

  \begin{example}
    Let $G$ be the graph from Figure \ref{img:refinement} and let  $\mu =
    \begin{smallmatrix}\{b\} \{a, c, d\} \\ \{d\} \{a, b, c\}
    \end{smallmatrix}$ be a partition mapping on its vertices. Let $W=\{a, c,
    d\}$ (therefore $\mu(W)=\{a,b,c\}$) be the cell with respect to which the
    split operation is being performed. We have the domain partition being
    split into $\{b\}\{a,c\}\{d\}$ since $\operatorname{deg}(a, W) = 1$,
    $\operatorname{deg}(c, W) = 1$ and $\operatorname{deg}(d, W) = 2$. Since
    $\operatorname{deg}(a, \mu(W)) = 1$, $\operatorname{deg}(b, \mu(W)) = 2$
    and $\operatorname{deg}(c, \mu(W)) = 1$ it means that $\mu'(\{a,c\}) =
    \{a,c\}$ and $\mu'(\{d\}) = \{b\}$, i.e. $s(G, \mu, W) = \mu' =
    \begin{smallmatrix}\{b\} \{a, c\} \{d\} \\ \{d\} \{a, c\} \{b\}
    \end{smallmatrix}$.

		\begin{figure}[H]
		\begin{center}
			\begin{tikzpicture}
  \begin{scope}[transform shape, scale=0.85]
    \node (a) at (-1, 1) [fill, circle, label=180:$a$] {};
    \node (b) at (0, 2) [fill, circle, label=0:$b$] {};
    \node (c) at (1, 1) [fill, circle, label=0:$c$] {};
    \node (d) at (0, 0) [fill, circle, label=0:$d$] {};

    \path [-] (a) edge (b);
    \path [-] (b) edge (c);
    \path [-] (c) edge (d);
    \path [-] (d) edge (a);
    \path [-] (b) edge (d);
  \end{scope}
\end{tikzpicture}
		\end{center}
    \caption{An example graph.}
		\label{img:refinement}
		\end{figure}
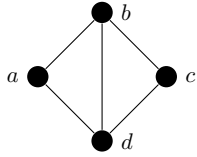
  \end{example}

  We define the \emph{refinement} operation $r(G, \mu)$ as the repeated,
  exhaustive application of $s$, where at each step $W$ is a cell of the domain
  partition of the partition mapping at that step. We do not specify the order
  in which these operations are performed, and we assume that refinement halts as
  soon as the resulting partition mapping is partial. We define $P(G, \nu)$ as
  follows:
  \begin{itemize}
    \item $P(G, []) = r(G, \mu)$ where $\mu(V)=V$
    \item $P(G, [\nu, u]) = r(G, i(P(G, \nu), u))$
  \end{itemize}

  \begin{example}
    Figure \ref{img:refinement_partition} shows the search tree with the partition mappings
    associated with each node.

		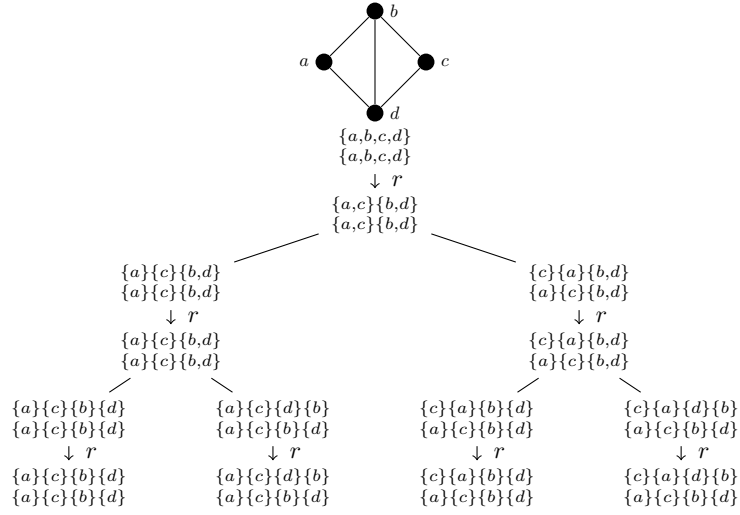
\begin{figure}[H]
		\begin{center}
			\begin{tikzpicture}
  \begin{scope}[transform shape, scale=0.9]
    \begin{scope}[shift={(0,1.5)}, transform shape, scale=0.75]
      \node (a) at (-1, 1) [fill, circle, label=180:$a$] {};
      \node (b) at (0, 2) [fill, circle, label=0:$b$] {};
      \node (c) at (1, 1) [fill, circle, label=0:$c$] {};
      \node (d) at (0, 0) [fill, circle, label=0:$d$] {};

      \path [-] (a) edge (b);
      \path [-] (b) edge (c);
      \path [-] (c) edge (d);
      \path [-] (d) edge (a);
      \path [-] (b) edge (d);
    \end{scope}
    \node (0a) at (0, 1) [] {$\begin{smallmatrix}\{a,b,c,d\} \\ \{a,b,c,d\}\end{smallmatrix}$};
    \node (0b) at (0, 0) [] {$\begin{smallmatrix}\{a,c\}\{b,d\} \\ \{a,c\}\{b,d\}\end{smallmatrix}$};

    \node (1a) at (-3, -1) [] {$\begin{smallmatrix}\{a\}\{c\}\{b,d\} \\ \{a\}\{c\}\{b,d\}\end{smallmatrix}$};
    \node (1b) at (-3, -2) [] {$\begin{smallmatrix}\{a\}\{c\}\{b,d\} \\ \{a\}\{c\}\{b,d\}\end{smallmatrix}$};
    \node (2a) at (3, -1) [] {$\begin{smallmatrix}\{c\}\{a\}\{b,d\} \\ \{a\}\{c\}\{b,d\}\end{smallmatrix}$};
    \node (2b) at (3, -2) [] {$\begin{smallmatrix}\{c\}\{a\}\{b,d\} \\ \{a\}\{c\}\{b,d\}\end{smallmatrix}$};

    \node (3a) at (-4.5, -3) [] {$\begin{smallmatrix}\{a\}\{c\}\{b\}\{d\} \\ \{a\}\{c\}\{b\}\{d\}\end{smallmatrix}$};
    \node (3b) at (-4.5, -4) [] {$\begin{smallmatrix}\{a\}\{c\}\{b\}\{d\} \\ \{a\}\{c\}\{b\}\{d\}\end{smallmatrix}$};

    \node (4a) at (-1.5, -3) [] {$\begin{smallmatrix}\{a\}\{c\}\{d\}\{b\} \\ \{a\}\{c\}\{b\}\{d\}\end{smallmatrix}$};
    \node (4b) at (-1.5, -4) [] {$\begin{smallmatrix}\{a\}\{c\}\{d\}\{b\} \\ \{a\}\{c\}\{b\}\{d\}\end{smallmatrix}$};

    \node (5a) at (1.5, -3) [] {$\begin{smallmatrix}\{c\}\{a\}\{b\}\{d\} \\ \{a\}\{c\}\{b\}\{d\}\end{smallmatrix}$};
    \node (5b) at (1.5, -4) [] {$\begin{smallmatrix}\{c\}\{a\}\{b\}\{d\} \\ \{a\}\{c\}\{b\}\{d\}\end{smallmatrix}$};

    \node (6a) at (4.5, -3) [] {$\begin{smallmatrix}\{c\}\{a\}\{d\}\{b\} \\ \{a\}\{c\}\{b\}\{d\}\end{smallmatrix}$};
    \node (6b) at (4.5, -4) [] {$\begin{smallmatrix}\{c\}\{a\}\{d\}\{b\} \\ \{a\}\{c\}\{b\}\{d\}\end{smallmatrix}$};

		\path [->] (0a) edge node[label=0:$r$]{} (0b);
    \path [-] (0b) edge (1a);
    \path [-] (0b) edge (2a);
    \path [->] (1a) edge node[label=0:$r$]{} (1b);
    \path [->] (2a) edge node[label=0:$r$]{} (2b);
    \path [-] (1b) edge (3a);
    \path [-] (1b) edge (4a);
    \path [-] (2b) edge (5a);
    \path [-] (2b) edge (6a);
    \path [->] (3a) edge node[label=0:$r$]{} (3b);
    \path [->] (4a) edge node[label=0:$r$]{} (4b);
    \path [->] (5a) edge node[label=0:$r$]{} (5b);
    \path [->] (6a) edge node[label=0:$r$]{} (6b);
  \end{scope}
\end{tikzpicture}
		\end{center}
    \caption{Search tree with partition mappings.}
		\label{img:refinement_partition}
		\end{figure}
  \end{example}

  \begin{example}
    Figure \ref{img:partial} shows a part of the search tree for the given
    graph. The left branch shows the steps for
    computing $r$ where the resulting partition mapping is bijective,
    while the right branch shows the steps where the resulting partition
    mapping is partial.

		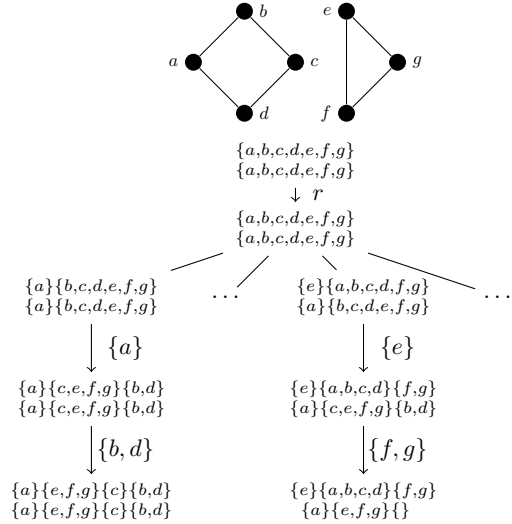
\begin{figure}[H]
		\begin{center}
			\begin{tikzpicture}
  \begin{scope}[transform shape, scale=0.9]
    \begin{scope}[shift={(0,1.7)}, transform shape, scale=0.75]
      \node (a) at (-2, 1) [fill, circle, label=180:$a$] {};
      \node (b) at (-1, 2) [fill, circle, label=0:$b$] {};
      \node (c) at (0, 1) [fill, circle, label=0:$c$] {};
      \node (d) at (-1, 0) [fill, circle, label=0:$d$] {};

      \node (e) at (1, 2) [fill, circle, label=180:$e$] {};
      \node (f) at (1, 0) [fill, circle, label=180:$f$] {};
      \node (g) at (2, 1) [fill, circle, label=0:$g$] {};

      \path [-] (a) edge (b);
      \path [-] (b) edge (c);
      \path [-] (c) edge (d);
      \path [-] (d) edge (a);

      \path [-] (e) edge (f);
      \path [-] (f) edge (g);
      \path [-] (g) edge (e);
    \end{scope}

    \node (0a) at (0, 1) [] {$\begin{smallmatrix}\{a,b,c,d,e,f,g\} \\ \{a,b,c,d,e,f,g\}\end{smallmatrix}$};
    \node (0b) at (0, 0) [] {$\begin{smallmatrix}\{a,b,c,d,e,f,g\} \\ \{a,b,c,d,e,f,g\}\end{smallmatrix}$};

    \node (1a) at (-3, -1) [] {$\begin{smallmatrix}\{a\}\{b,c,d,e,f,g\} \\ \{a\}\{b,c,d,e,f,g\}\end{smallmatrix}$};
    \node (1b) at (-3, -2.5) [] {$\begin{smallmatrix}\{a\}\{c,e,f,g\}\{b,d\} \\ \{a\}\{c,e,f,g\}\{b,d\}\end{smallmatrix}$};
    \node (1c) at (-3, -4) [] {$\begin{smallmatrix}\{a\}\{e,f,g\}\{c\}\{b,d\} \\ \{a\}\{e,f,g\}\{c\}\{b,d\}\end{smallmatrix}$};
    \node (dots1) at (-1, -1) [] {$\dots$};

    \node (2a) at (1, -1) [] {$\begin{smallmatrix}\{e\}\{a,b,c,d,f,g\} \\ \{a\}\{b,c,d,e,f,g\}\end{smallmatrix}$};
    \node (2b) at (1, -2.5) [] {$\begin{smallmatrix}\{e\}\{a,b,c,d\}\{f,g\} \\ \{a\}\{c,e,f,g\}\{b,d\}\end{smallmatrix}$};
    \node (2c) at (1, -4) [] {$\begin{smallmatrix}\{e\}\{a,b,c,d\}\{f,g\} \\ \{a\}\{e,f,g\}\{\}\ \ \end{smallmatrix}$};
    \node (dots2) at (3, -1) [] {$\dots$};

		\path [->] (0a) edge node[label=0:$r$]{} (0b);
    \path [-] (0b) edge (1a);
    \path [-] (0b) edge (2a);
    \path [-] (0b) edge (dots1);
    \path [-] (0b) edge (dots2);
    \path [->] (1a) edge (1b);
    \path [->] (1b) edge (1c);
    \path [->] (2a) edge (2b);
    \path [->] (2b) edge (2c);

    \node (label1a) at (-2.5, -1.75) [] {$\{a\}$};
    \node (label1b) at (-2.5, -3.25) [] {$\{b,d\}$};
    \node (label2a) at (1.5, -1.75) [] {$\{e\}$};
    \node (label2b) at (1.5, -3.25) [] {$\{f,g\}$};
  \end{scope}
\end{tikzpicture}
		\end{center}
    \caption{Part of the search tree with partition mappings.}
		\label{img:partial}
		\end{figure}
  \end{example}

  Note the simple observation that $\sigma_{P(G, \nu)}(u)$ is preserved for all
  individualized vertices $u$ when moving down the tree, as long as the
  partition mapping is bijective. We formalize this result in the following
  lemma.

  \begin{lemma}
    \label{lemma:sigma}
    Let $[\nu, u]$ be a node. For all $\nu' \in \mathcal{T}(G, [\nu, u])$ such
    that $P(G, \nu')$ is bijective, it holds that $\sigma_{P(G, \nu')} (u) =
    \sigma_{P(G, \nu)}(u)$.
  \end{lemma}
  \begin{proof}
    Refinement may not change the mapping of singleton cells if the resulting
    partition mapping is bijective. For any $\nu'$ which is a proper descendant
    of $[\nu, u]$ with a bijective partition mapping, clearly $\sigma_{P(G,
    \nu')}(u) = \sigma_{P(G, [\nu, u])}(u)$ must hold since $\{u\}$ is a
    singleton in $P(G, [\nu, u])$. Since, by definition, $\{u\}$ maps to
    $\{\sigma_{P(G, \nu)}(u)\}$ after individualization, we have $\sigma_{P(G,
    [\nu, u])}(u) = \sigma_{P(G, \nu)}(u)$. Therefore, $\sigma_{P(G, \nu')}(u)
    = \sigma_{P(G, \nu)}(u)$ for all $\nu' \in \mathcal{T}(G, [\nu, u])$ with a
    bijective partition mapping.
  \end{proof}

  The automorphisms are discovered at the leaves of the search tree whose
  partition mapping is bijective. The following lemma justifies this.

  \begin{lemma}
    \label{lemma:leaf_aut}
    Let $\nu$ be a leaf and $\mu = P(G, \nu)$ be its bijective partition
    mapping. Then $\sigma_\mu \in \operatorname{Aut}(G)$.
  \end{lemma}
  \begin{proof}
    Since $\nu$ is a leaf with a bijective partition mapping, the domain
    partition of $\mu$ is $\{ \{u\} \mid u \text{ is a vertex of } G \}$. Based
    on the definition of $r$, any additional application of $s$ to $\mu$ does
    not further refine it, i.e. $s(G, \mu, W) = \mu$ for any cell $W$ of the
    domain partition of $\mu$. This means that $\operatorname{deg}(u, \{v\}) =
    \operatorname{deg}(\sigma_\mu(u), \{\sigma_\mu(v)\})$ for any two vertices
    $u$ and $v$ in $G$. Therefore, $u$ and $v$ are adjacent in $G$ if and only
    if $\sigma_\mu(u)$ and $\sigma_\mu(v)$ are adjacent in $G$, meaning that
    $\sigma_\mu$ is an automorphism of $G$.
  \end{proof}

  Let $\nu^\sigma=[\sigma(u_1), \dots, \sigma(u_k)]$ denote the
  image of node $\nu=[u_1, \dots, u_k]$ under permutation
  $\sigma$. We introduce a pruning operation $\Pi([\nu, u],
  \sigma)$, allowing the subtree with root $[\nu, u']$ to be
  pruned if $u < u'$ and $[\nu, u]^\sigma = [\nu, u']$ for some
  automorphism $\sigma$ of $G$. This is essentially a pruning
  operation described by McKay and Piperno \cite{McKayPiperno},
  adapted to our construction. The operation is useful because
  the automorphisms of the pruned subtree do not expand the
  generated group, meaning the search can skip the entire
  subtree. The formal justification will come later, as Corollary
  \ref{cor:pruning}.

  \begin{example}
    Figure \ref{img:refinement_prune} shows the search tree from Figure
    \ref{img:refinement_partition} with the parts of the search tree which may
    be pruned being highlighted.

		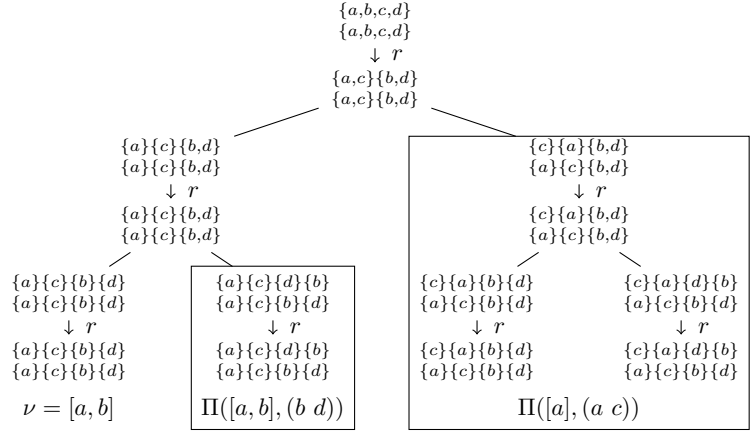
\begin{figure}[H]
		\begin{center}
			\begin{tikzpicture}
  \begin{scope}[transform shape, scale=0.9]
    \node (0a) at (0, 1) [] {$\begin{smallmatrix}\{a,b,c,d\} \\ \{a,b,c,d\}\end{smallmatrix}$};
    \node (0b) at (0, 0) [] {$\begin{smallmatrix}\{a,c\}\{b,d\} \\ \{a,c\}\{b,d\}\end{smallmatrix}$};

    \node (1a) at (-3, -1) [] {$\begin{smallmatrix}\{a\}\{c\}\{b,d\} \\ \{a\}\{c\}\{b,d\}\end{smallmatrix}$};
    \node (1b) at (-3, -2) [] {$\begin{smallmatrix}\{a\}\{c\}\{b,d\} \\ \{a\}\{c\}\{b,d\}\end{smallmatrix}$};
    \node (2a) at (3, -1) [] {$\begin{smallmatrix}\{c\}\{a\}\{b,d\} \\ \{a\}\{c\}\{b,d\}\end{smallmatrix}$};
    \node (2b) at (3, -2) [] {$\begin{smallmatrix}\{c\}\{a\}\{b,d\} \\ \{a\}\{c\}\{b,d\}\end{smallmatrix}$};

    \node (3a) at (-4.5, -3) [] {$\begin{smallmatrix}\{a\}\{c\}\{b\}\{d\} \\ \{a\}\{c\}\{b\}\{d\}\end{smallmatrix}$};
    \node (3b) at (-4.5, -4) [] {$\begin{smallmatrix}\{a\}\{c\}\{b\}\{d\} \\ \{a\}\{c\}\{b\}\{d\}\end{smallmatrix}$};

    \node (4a) at (-1.5, -3) [] {$\begin{smallmatrix}\{a\}\{c\}\{d\}\{b\} \\ \{a\}\{c\}\{b\}\{d\}\end{smallmatrix}$};
    \node (4b) at (-1.5, -4) [] {$\begin{smallmatrix}\{a\}\{c\}\{d\}\{b\} \\ \{a\}\{c\}\{b\}\{d\}\end{smallmatrix}$};

    \node (5a) at (1.5, -3) [] {$\begin{smallmatrix}\{c\}\{a\}\{b\}\{d\} \\ \{a\}\{c\}\{b\}\{d\}\end{smallmatrix}$};
    \node (5b) at (1.5, -4) [] {$\begin{smallmatrix}\{c\}\{a\}\{b\}\{d\} \\ \{a\}\{c\}\{b\}\{d\}\end{smallmatrix}$};

    \node (6a) at (4.5, -3) [] {$\begin{smallmatrix}\{c\}\{a\}\{d\}\{b\} \\ \{a\}\{c\}\{b\}\{d\}\end{smallmatrix}$};
    \node (6b) at (4.5, -4) [] {$\begin{smallmatrix}\{c\}\{a\}\{d\}\{b\} \\ \{a\}\{c\}\{b\}\{d\}\end{smallmatrix}$};

		\path [->] (0a) edge node[label=0:$r$]{} (0b);
    \path [-] (0b) edge (1a);
    \path [-] (0b) edge (2a);
    \path [->] (1a) edge node[label=0:$r$]{} (1b);
    \path [->] (2a) edge node[label=0:$r$]{} (2b);
    \path [-] (1b) edge (3a);
    \path [-] (1b) edge (4a);
    \path [-] (2b) edge (5a);
    \path [-] (2b) edge (6a);
    \path [->] (3a) edge node[label=0:$r$]{} (3b);
    \path [->] (4a) edge node[label=0:$r$]{} (4b);
    \path [->] (5a) edge node[label=0:$r$]{} (5b);
    \path [->] (6a) edge node[label=0:$r$]{} (6b);

    \draw (-2.7, -2.6) rectangle (-0.3, -5);
    \draw (0.5, -0.69) rectangle (5.5, -5);
    \node (node) at (-4.5, -4.7) [] {$\nu = [a, b]$};
    \node (prune1) at (-1.5, -4.7) [] {$\Pi([a, b], (b\ d))$};
    \node (prune2) at (3, -4.7) [] {$\Pi([a], (a\ c))$};
  \end{scope}
\end{tikzpicture}
		\end{center}
    \caption{Search tree with pruned subtrees highlighted.}
		\label{img:refinement_prune}
		\end{figure}
  \end{example}

  The automorphisms used for pruning may be discovered during the execution of
  the algorithm. The following lemma describes how the automorphisms discovered
  at two leaves of the search tree may be used to prune the subtree containing
  one of the leaves.

  \begin{lemma}
    \label{lemma:pruning_discovery}
    Let $\nu_1 < \nu_2$ be two leaves of the search tree where automorphisms
    $\sigma_1$ and $\sigma_2$ are discovered, respectively. Let $\sigma =
    \sigma_2^{-1} \circ \sigma_1$ and let $l$ be the length of the longest
    common ancestor of $\nu_1$ and $\nu_2$. Then the subtree rooted at
    $[\nu_2]_{l+1}$ may be pruned using the pruning operation
    $\Pi([\nu_1]_{l+1}, \sigma)$.
  \end{lemma}
  \begin{proof}
    Let $\nu = [\nu_1]_l = [\nu_2]_l$ and $\mu = P(G, \nu)$. By Lemma \ref{lemma:sigma},
    $\sigma_1(u) = \sigma_\mu(u) = \sigma_2(u)$ for any $u \in \nu$. Let $[\nu,
    u_1] = [\nu_1]_{l+1}$ and $[\nu, u_2] = [\nu_2]_{l+1}$. Since $u_1, u_2 \in
    T(G, \nu)$, we have $\sigma_1(u_1) = \sigma_2(u_2)$. Therefore $[\nu,
    u_1]^\sigma = [\nu, u_2]$.
  \end{proof}

  In practice, the first discovered leaf $\nu_0$ with a bijective mapping is
  always used in combination with a newly discovered leaf for such pruning.

  Automorphisms, if known in advance, may be provided before starting the
  traversal of the search tree, leading to more effective pruning. We rely on
  this fact for constructing the certificate. An initial execution of the
  algorithm is used to find the generators of the automorphism group.
  Afterwards, a second traversal is executed in order to construct the
  certificate based on the now pruned tree. This is done so that the
  certificate size is further optimized, as this is computationally the
  heaviest part of the algorithm.

	\subsection{Certificate construction}

  After an initial traversal of the tree, a set of automorphism group
  generators $\{a_1, \dots, a_m\}$ is obtained. The main idea for constructing
  the certificate is assigning a monomial of the form $M_\nu = \prod_{i=1}^k
  x_{u_i, v_i}$ to each node $\nu=[u_1, \dots, u_k]$, where $v_i = \sigma_{P(G,
  [\nu]_{i-1})}(u_i)$ for $1 \leq i \leq k$. If $M_\nu$ can be derived from
  $\mathcal{A}(G, a_1, \dots, a_m)$, this means that $M_\nu(X)$ vanishes
  whenever $\mathcal{A}(G, a_1, \dots, a_m)(X)$ vanishes. In that case, at
  least one of $x_{u_i, v_i}$ is zero, meaning that no permutation matrix of a
  permutation which maps $u_1$ to $v_1$, $u_2$ to $v_2$, \dots, $u_k$ to $v_k$
  is a zero of $\mathcal{A}(G, a_1, \dots, a_m)$. This shows that any
  automorphism found in the subtree of $\nu$ must be contained in $\langle a_1,
  \dots, a_m \rangle$.

	\begin{example}
    Figure \ref{img:monomials} shows the monomials associated with each node of
    the tree given in Figure \ref{img:refinement_partition}.

		\begin{figure}[H]
		\begin{center}
			\begin{tikzpicture}
  \begin{scope}
    \node (0) at (0, 0) [] {$1$};
    \node (1) at (-2, -1) [] {$x_{a, a}$};
    \node (2) at (2, -1) [] {$x_{c, a}$};
    \node (3) at (-3, -2) [] {$x_{a, a} x_{b, b}$};
    \node (4) at (-1, -2) [] {$x_{a, a} x_{d, b}$};
    \node (5) at (1, -2) [] {$x_{c, a} x_{b, b}$};
    \node (6) at (3, -2) [] {$x_{c, a} x_{d, b}$};
  \end{scope}
  \begin{scope}
    \path [-] (0) edge (1);
    \path [-] (0) edge (2);
    \path [-] (1) edge (3);
    \path [-] (1) edge (4);
    \path [-] (2) edge (5);
    \path [-] (2) edge (6);
  \end{scope}
\end{tikzpicture}
		\end{center}
		\caption{A search tree with its associated node monomials.}
		\label{img:monomials}
		\end{figure}
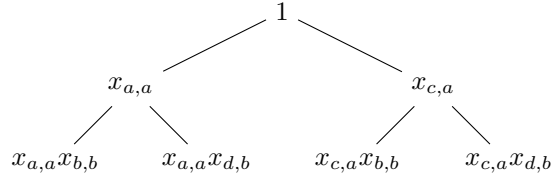
	\end{example}

  The construction of the certificate rests on the following lemma, which shows
  a clear connection between the partition mappings and the polynomials
  representing them. We present a proof sketch of the lemma here, while the
  full proof is given in Appendix \ref{app:proof}.

	\begin{lemma}
		\label{lemma:monorefine}
    Let $\nu$ be a node of the search tree with partition mapping $\mu$.
    Polynomials of the form $(-1 + \sum_{v \in \mu[u]} x_{u, v}) M_\nu$ may be
    derived from $\mathcal{A}(G, a_1, \dots, a_m)$ for each vertex $u$, $(-1 + \sum_{u \in \mu^{-1}[v]} x_{u, v})
    M_\nu$ may be derived for each vertex $v$, and $x_{u, v} M_\nu$ may be
    derived for each $u$ and $v$ such that $v \notin \mu[u]$.
	\end{lemma}
  \begin{proof}[Proof sketch]
    Consider the partition mapping $\mu$ of node $\nu$. Starting from the root
    of the search tree, we may interpret the steps taken to calculate $\mu$ as
    a chain of individualization and refinement steps. Specifically, we
    consider the chain $\mu_0 \overset{i}{\mapsto} \mu_1
    \overset{s}{\mapsto} \dots \overset{s}{\mapsto} \mu_j
    \overset{i}{\mapsto} \mu_{j+1} \overset{s}{\mapsto} \dots
    \overset{s}{\mapsto} \mu$. We will show, using induction on the
    number of refinement and individualization steps, that the appropriate
    polynomials are derivable for each of the partition mappings in the given
    chain. We will use the notation $p = g_1 \underline{p_1} + \dots + g_m
    \underline{p_m}$ to denote that $p_1, \dots, p_m$ are assumed to be
    derivable and that the Ideal rule may be applied with coefficients $g_1,
    \dots, g_m$ to derive $p$.

    \begin{induction}

    \item[Base case] At the root of the search tree $\nu = []$ the
      bijective partition mapping maps the entire set of vertices to itself and no
      refinement or individualization steps have been performed. The
      $\operatorname{RowS}$ and $\operatorname{ColS}$ polynomials given in the
      system $\mathcal{A}(G, a_1, \dots, a_m)$ are exactly the polynomials from
      the statement and they are derivable by the Tautology rule.

    \item[Induction hypothesis] The polynomials $m_{i, u} = (-1 + \sum_{v \in
      \mu_i[u]} x_{u, v}) M_{\nu}$ and $w_{i, v} = (-1 + \sum_{u \in
      \mu_i^{-1}[v]} x_{u, v}) M_{\nu}$ are derivable for each vertex $u$ and
      $v$ in the bijective partition mapping $\mu_i:\mathbb{U}\to\mathbb{V}$ of node
      $\nu$, and the polynomials $e_{i, u, v} = x_{u, v} M_{\nu}$ are derivable
      for each pair of vertices $u$ and $v$ such that $v \notin \mu_i[u]$.

    \item[Induction step] We consider two different induction steps, one for
      individualization and one for refinement.

      (1) Assume $\mu_{i+1}$ is the partition mapping of $\nu' = [\nu, u_k]$ produced by
      individualizing vertex $u_k$ to vertex $v_k$ from $\mu_i$. We may derive
      the existing exclusion $e_{i+1, u, v}$ by multiplying $e_{i, u, v}$ by
      $x_{u_k, v_k}$, or a new one produced by individualization (between $u_k$
      and the rest of $V$, and $v_k$ and the rest of $U$) using RowNZ and ColNZ
      polynomials. Finally, we may derive $m_{i+1, u}$ for any vertex $u$ and
      $w_{i+1, v}$ for any vertex $v$ by subtracting the derived exclusion
      polynomials from $x_{u_k, v_k} m_{i, u}$ and $x_{u_k, v_k} w_{i, v}$
      respectively.

      (2) Assume $\mu_{i+1}$ is produced by splitting $\mu_i$ according to some
      cell $W \in \mathbb{U}$. We show the following identities, which connect
      the Iso family of polynomials to degree counts:
      \begin{align*}
        \sum_{v \in \mu_i(W)} (AX)_{u, v} &= \sum_{t \in n_G(u)} \sum_{v \in \mu_i(W)} x_{t, v} \\
        \sum_{v \in \mu_i(W)} (XA)_{u, v} &= \sum_{t=0}^{n-1} x_{u, t} \operatorname{deg}(t, \mu_i(W))
      \end{align*}

    Let $d_{u, W} = \operatorname{deg}(u, W) - \sum_{v \in \mu_i[u]} x_{u, v}
      \operatorname{deg}(v, \mu_i(W))$. For any vertex $u$ we show that $d_{u,
      W} M_{\nu}$ is derivable, by a combination of the aforementioned
      identities and the induction hypothesis. We then show that $ce_{i+1, u,
      v}$ is derivable, where $c=\operatorname{deg}(u, W) -
      \operatorname{deg}(v, \mu_i(W))$. This means that $e_{i+1, u, v}$ is
      derivable when $c$ is invertible in $\mathbb{F}_p$, and since $|c| < n <
      p$, that is when $c \neq 0$. In other words, $e_{i+1, u, v}$ is derivable
      if $u$ and $v$ get separated by refinement. Finally, we may derive
      $m_{i+1, u}$ for any vertex $u$ and $w_{i+1, v}$ for any vertex $v$ by
      subtracting the derived exclusion polynomials from $m_{i, u}$ and $w_{i,
      v}$ respectively. \qedhere
    \end{induction}
	\end{proof}

	\begin{example}
		Let $\mu = \begin{smallmatrix} \{0, 3\} \{1, 2, 4\} \\ \{0, 3\} \{1, 2, 4\}
		\end{smallmatrix}$ be the bijective partition mapping at node $\nu = []$. Then $x_{3,
		0} + x_{3, 3} - 1$ is derivable since $\mu[3]=\{0, 3\}$, $x_{1, 2} + x_{2, 2} +
		x_{4, 2} - 1$ is derivable since $\mu^{-1}[2] = \{1, 2, 4\}$, and $x_{0, 4}$ is
		derivable since $4 \notin \mu[0]$.
	\end{example}

  We now present the main result of this section. The theorem states that
  the algorithm is always able to construct the certificate of infeasibility of
  $\mathcal{A}(G, a_1, \dots, a_m)$.

	\begin{theorem}
		\label{thm:monotree}
    Let $a_1, \dots, a_m$ be the automorphisms found at
    the leaves of the pruned search tree
    $\mathcal{T}(G)$ and the automorphisms used in the
    pruning operations $\Pi$. Let $\nu$ be a node of
    the pruned tree or the root of a pruned subtree.
    Then $M_\nu$ is derivable from $\mathcal{A}(G, a_1,
    \dots, a_m)$.
	\end{theorem}

	\begin{proof}
    Let $\nu = [u_1, \dots, u_k]$ and $v_i = \sigma_{P(G, [\nu]_{i-1})}(u_i)$ for $1
    \leq i \leq k$. We consider four different cases.

    (1) $\nu$ is a leaf with a bijective partition mapping $\mu =
    P(G, \nu)$. Let $u_{k+1}, \dots, u_n$ be all of the vertices
    that were not individualized, i.e. which do not appear in
    $\nu = [u_1, \dots, u_k]$, and let $v_{k+1} =
    \sigma_\mu(u_{k+1}), \dots, v_n=\sigma_\mu(u_n)$. From Lemma
    \ref{lemma:sigma} we also have $v_i = \sigma_\mu(u_i)$ for $1
    \leq i \leq k$. Based on Lemma \ref{lemma:monorefine} and the
    fact that $\mu$ is discrete, $(x_{u_j, v_j} - 1) M_\nu$ is
    derivable for any $k < j \leq n$. We show that
    $\prod_{i=1}^{n-1} x_{u_i, v_i}$ is derivable:

    \begin{align*}
      \prod_{i=1}^{n-1} x_{u_i, v_i} &= \prod_{i=1}^n x_{u_i, v_i} - \prod_{i=1}^n x_{u_i, v_i} + \prod_{i=1}^{n-1} x_{u_i, v_i} \\
      &= \prod_{i=1}^n x_{u_i, v_i} - (\prod_{\mathclap{j=k+1}}^n x_{u_j, v_j} - \prod_{\mathclap{j=k+1}}^{n-1} x_{u_j, v_j}) \prod_{i=1}^k x_{u_i, v_i} \\
      &= \prod_{i=1}^n x_{u_i, v_i} - \prod_{\mathclap{j=k+1}}^{n-1} x_{u_j, v_j} (x_{u_n, v_n} - 1) \prod_{i=1}^k x_{u_i, v_i} \\
      &= \underline{\operatorname{ExIso}_{\sigma_\mu}} - \prod_{\mathclap{j=k+1}}^{n-1} x_{u_j, v_j} \underline{(x_{u_n, v_n} - 1) M_\nu}
    \end{align*}
    Similarly, $\prod_{i=1}^{n-2} x_{u_i, v_i} =
    \underline{\prod_{i=1}^{n-1} x_{u_i, v_i}} -
    \prod_{j=k+1}^{n-2} x_{u_j, v_j} \underline{(x_{u_{n-1},
    v_{n-1}} - 1) M_\nu}$ is derivable, and this process may be
    repeated to show that $M_\nu = \prod_{i=1}^k x_{u_i, v_i}$ is
    derivable. Note that $\sigma_\mu$ is one of $a_1, \dots, a_m$
    since it is an automorphism found at the leaf $\nu$ of the
    search tree, due to Lemma \ref{lemma:leaf_aut}, which means
    that $\operatorname{ExIso}_{\sigma_\mu}$ is derivable.

    (2) $\nu = [u_1, \dots, u_k]$ is pruned by $\Pi([u_1, \dots,
    u_{k-1}, u'_k], \sigma)$. Since $u_k$ and $u'_k$ share the
    same cell in $\mu = P(G, [u_1, \dots, u_{k-1}])$, it holds
    that $\sigma_\mu(u_k) = \sigma_\mu(u'_k)$. We then have:

    \begin{align*}
      M_\nu &= M_{[u_1, \dots, u_{k-1}]} x_{u_k, \sigma_\mu(u_k)} \\
      &= M_{[u_1, \dots, u_{k-1}]} x_{\sigma(u'_k), \sigma_\mu(u'_k)} \\
      &= M_{[u_1, \dots, u_{k-1}]} x_{u'_k, \sigma_\mu(u'_k)}^\sigma \\
      &= \underline{M_{[u_1, \dots, u_{k-1}, u'_k]}}^\sigma
    \end{align*}

    (3) $\nu$ is a leaf with a partial partition mapping $\mu$. There exists a
    cell $U$ such that $|U| \neq |\mu(U)|$. Due to Lemma \ref{lemma:monorefine}
    the polynomials $m_u = (-1 + \sum_{v \in \mu(U)} x_{u, v}) M_\nu$ and $w_v = (-1
    + \sum_{u \in U} x_{u, v}) M_\nu$ are derivable for any vertex $u \in U$ and
    $v \in \mu(U)$. Let
    $c = |U| - |\mu(U)|$ and note that $c$ is invertible in $\mathbb{F}_p$
    since $0 \leq |U|, |\mu(U)| \leq n < p$ and $c \neq 0$. We have:

    \begin{align*}
      & \sum_{\mathclap{v \in \mu(U)}} w_v - \sum_{\mathclap{u \in U}} m_u \\
      =& \sum_{\mathclap{v \in \mu(U)}} (-1 + \sum_{\mathclap{u \in U}} x_{u, v}) M_\nu - \sum_{\mathclap{u \in U}} (-1 + \sum_{\mathclap{v \in \mu(U)}} x_{u, v}) M_\nu \\
      =& -|\mu(U)|M_\nu + \smashoperator[l]{\sum_{v \in \mu(U)}} \sum_{u \in U} x_{u, v} M_\nu + |U|M_\nu - \sum_{\mathclap{u \in U}} \smashoperator[r]{\sum_{v \in \mu(U)}} x_{u, v} M_\nu \\
      =& -|\mu(U)|M_\nu + |U|M_\nu \\
      =& c M_\nu
    \end{align*}
    From this, we can derive $M_\nu$ by multiplying with $c^{-1}$.

    (4) $\nu$ is an inner node. Let $\nu_1, \dots, \nu_l$ be its child nodes
    produced by individualizing $u'_1, \dots, u'_l$ respectively, and let $v'$
    be the vertex being mapped to. Due to Lemma \ref{lemma:monorefine}, the
    polynomial $w = (-1 + \sum_{i=1}^l x_{u'_i, v'}) M_\nu$ is derivable. We can
    derive $M_\nu$ as follows:

    \begin{align*}
      M_\nu &= M_\nu - \sum_{j=1}^l M_{\nu_j} + \sum_{j=1}^l M_{\nu_j} \\
      &= M_\nu - \sum_{j=1}^l x_{u'_j, v'} M_\nu + \sum_{j=1}^l M_{\nu_j} \\
      &= -M_\nu (-1 + \sum_{i=1}^l x_{u'_i, v'}) + \sum_{j=1}^l M_{\nu_j} \\
      &= -\underline{w} + \sum_{j=1}^l \underline{M_{\nu_j}}
    \end{align*}
	\end{proof}

	\begin{example}
		The example in Figure \ref{img:proof} shows how the monomials corresponding to
		nodes from Figure \ref{img:refinement_partition} may be derived, following the
    construction from the proof of Theorem \ref{thm:monotree}.
		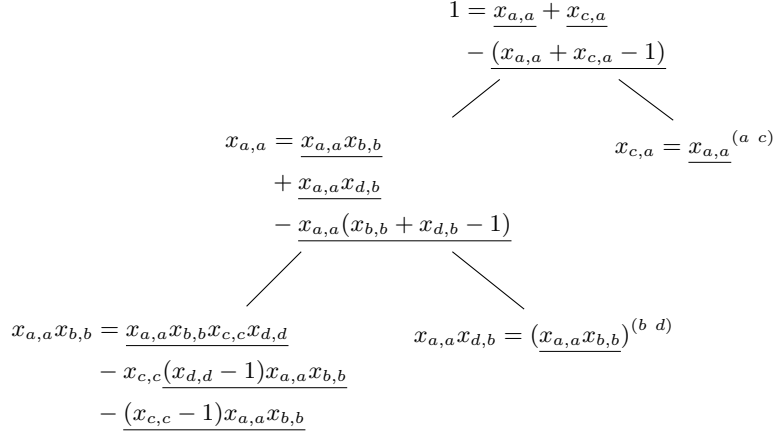
\begin{figure}[H]
		\begin{center}
			\begin{tikzpicture}
  \begin{scope}
    \node (0) at (0, 0) [] {$\small\begin{aligned}1 &= \underline{x_{a, a}} + \underline{x_{c, a}} \\ &- \underline{(x_{a, a} + x_{c, a} - 1)}\end{aligned}$};
    \node (1) at (-2.5, -2) [] {$\small\begin{aligned}x_{a, a} &= \underline{x_{a, a} x_{b, b}} \\ &+ \underline{x_{a, a} x_{d, b}} \\ &- \underline{x_{a, a}(x_{b, b} + x_{d, b} - 1)}\end{aligned}$};
    \node (2) at (2, -1.5) [] {$\small\begin{aligned}x_{c, a} &= \underline{x_{a, a}}^{(a\ c)}\end{aligned}$};
    \node (3) at (-5, -4.5) [] {$\small\begin{aligned}x_{a, a} x_{b, b} &= \underline{x_{a, a} x_{b, b} x_{c, c} x_{d, d}} \\ &- x_{c, c} \underline{(x_{d, d} - 1) x_{a, a} x_{b, b}} \\ &- \underline{(x_{c, c} - 1) x_{a, a} x_{b, b}}\end{aligned}$};
    \node (4) at (0, -4) [] {$\small\begin{aligned}x_{a, a} x_{d, b} &= (\underline{x_{a, a} x_{b, b}})^{(b\ d)}\end{aligned}$};
  \end{scope}
  \begin{scope}
    \path [-] (0) edge (1);
    \path [-] (0) edge (2);
    \path [-] (1) edge (3);
    \path [-] (1) edge (4);
  \end{scope}
\end{tikzpicture}
		\end{center}
      \caption{A tree showing how each monomial may be derived from
      $\mathcal{A}(G, a_1, \dots, a_m)$. The underlined polynomials are assumed
      to be derivable, either as part of the set $\mathcal{A}(G, a_1, \dots,
      a_m)$ using the Tautology rule, through Lemma \ref{lemma:monorefine}, or
      from the derivation in the child nodes.}
		\label{img:proof}
		\end{figure}
	\end{example}

  The following corollary shows the completeness of the construction and
  provides justification for the pruning operation $\Pi$.

  \begin{corollary}
    \label{cor:pruning}
    Let $a_1, \dots, a_m$ be the automorphisms found at the leaves of the
    pruned search tree $\mathcal{T}(G)$ and the automorphisms used in the
    pruning operations $\Pi$. Then $1$ is derivable from $\mathcal{A}(G, a_1,
    \dots, a_m)$ and $\operatorname{Aut}(G) = \langle a_1, \dots,
    a_m \rangle$.
  \end{corollary}
  \begin{proof}
    As a consequence of Theorem \ref{thm:monotree}, the root monomial $M_{[]} =
    1$ is derivable from $\mathcal{A}(G, a_1, \dots, a_m)$. By Corollary
    \ref{cor:soundness}, $\mathcal{A}(G, a_1, \dots, a_m)$ has no common
    zero, i.e. it is infeasible, and by Theorem \ref{thm:system_generators},
    $\mathcal{A}(G, a_1, \dots, a_m)$ is infeasible if and only if
    $\operatorname{Aut}(G) = \langle a_1, \dots, a_m \rangle$.
  \end{proof}

	\section{Implementation and results}

  The implementation consists of three components. The \emph{solver} implements
  the initial search phase of the algorithm described in Section
  \ref{sec:construction}. It takes a graph in the DIMACS format as input, and
  outputs a generating set of its automorphism group. The \emph{prover}
  implements the proof generation phase of the algorithm, executing the same
  search tree traversal, this time taking into account the already discovered
  automorphisms to prune the tree as early as possible. It takes the graph and
  the automorphisms as inputs, and produces a proof within the described proof
  system, certifying that the input automorphisms generate the entirety of the
  automorphism group. The \emph{verifier} takes this proof as input alongside
  the graph and the list of automorphisms, and checks the proof step by step in
  order to verify its correctness.

  The burden of trust is put entirely on the implementation of the verifier.
  Its correctness relies on the claims of soundness of the proof system,
  specifically Theorems \ref{thm:soundness} (and its corollary),
  \ref{thm:system_generators} and \ref{thm:closed}.  Therefore, by putting
  trust in the verifier, one also implicitly puts trust in the validity of
  those statements.

	\subsection{Implementation details}

  A prototype implementation of the three components has been developed in C++%
  \footnote{The source code is available at
  \url{https://github.com/idrecun/graph-aut-cert}}. The solver is implemented
  as a modification of the McKay-Piperno algorithm \cite{McKayPiperno}. Only the automorphism
  discovery aspect of the algorithm is implemented, and the modified pruning
  operations based on partition mappings were used instead of the typical
  pruning operations based on node invariants. More specifically, the algorithm
  maintains two colorings of graphs such that the mappings between the cells in
  corresponding positions match the partition mappings associated with the
  nodes of the search tree.

  The prover implements the same algorithm as the solver, but with added logic
  for proof generation. The proof format is a binary format consisting of a
  sequence of proof steps. Each proof step consists of a rule identifier,
  followed by its parameters. 

  The prover maintains a stack of proof contexts, each corresponding to a node
  along the current path in the search tree. The proof context holds the
  references to the steps whose resulting polynomial may be used in the
  later steps of the proof (specifically with the Ideal and Substitution
  rules). The proof is constructed following the steps outlined in the proof of
  Theorem \ref{thm:monotree} and its corresponding lemma. Because of this
  bottom-up construction, some of the steps may be redundant, i.e. not used to
  derive the final certificate, and are pruned before finalizing the proof
  file.

  The verifier first verifies that the input graph and the automorphisms are
  valid. It then parses the proof file and processes the proof steps one by
  one, calculating the resulting polynomial for the applied rule, and appending
  it to the array of derived polynomials. In the end, the verifier checks
  whether the last derived polynomial equals 1, and accepts or rejects the
  proof as a valid derivation of a Nullstellensatz certificate.

  The specific rules, their parameters and the resulting polynomials are
  described in the table below. The sequence of derived polynomials is denoted
  by $f$, and the sequence of input automorphisms is denoted by $a$.

  \begin{table}[!h]
    \begin{center}
    \begin{tabular}{ |l|l|l| }
      \hline
      \textbf{Rule} & \textbf{Parameters} & \textbf{Result} \\
      \hline
      Ideal & $0,n,i_1,\dots,i_n,p_1,\dots,p_n$ & $\sum_{j=1}^n f_{i_j} p_j$ \\
      \hline
      Substitution & $1,i,j$ & $f_i^{a_j}$ \\
      \hline
      Tautology (RowNZ) & $2,i,j,k$ & $\operatorname{RowNZ}_{i;j,k}$ \\
      \hline
      Tautology (ColNZ) & $3,i,j,k$ & $\operatorname{ColNZ}_{i,j;k}$ \\
      \hline
      Tautology (RowS) & $4,i$ & $\operatorname{RowS}_i$ \\
      \hline
      Tautology (ColS) & $5,j$ & $\operatorname{ColS}_j$ \\
      \hline
      Tautology (Iso) & $6,i,j$ & $\operatorname{Iso}_{i,j}$ \\
      \hline
      Tautology (ExIso) & $7,i$ & $\operatorname{ExIso}_{a_i}$ \\
      \hline
      Tautology (Bin) & $8,i,j$ & $\operatorname{Bin}_{i,j}$ \\
      \hline
    \end{tabular}
    \caption{Proof rule format specification.}
    \end{center}
  \end{table}

  The polynomials in the Ideal rule are represented by the sequence of
  monomials $n,m_1, \dots, m_n$ where $n$ is the sequence length. The monomials
  within a polynomial may be ordered arbitrarily. The monomials are represented
  by the sequence of variable-power pairs, followed by the coefficient $n, v_1,
  \dots, v_n, c$, where $n$ is the sequence length. Specifically, a
  variable-power pair $x_{i,j}^n$ is represented as a 4-byte integer where the
  12 highest order bits represent the row index $i$, the next
  12 bits represent the column index $j$, and the 8 lowest order bits represent
  the power $n$. These pairs must be ordered lexicographically within a
  monomial, with no two pairs having the same variable. The implementation
  works with a fixed finite field $\mathbb{F}_{4099}$, hence the coefficient
  must be a valid 2-byte integer in the range $[1, 4098]$. Proof length is a
  4-byte integer, rule identifiers are 1-byte integers, vertices and all
  sequence lengths are 2-byte integers, and polynomial and automorphism indices
  are 4-byte integers.

	\subsection{Experimental evaluation}

  The implementation was tested on a set of benchmark graphs from the
  \emph{bliss} collection\footnote{The benchmark set is available at
  \url{https://pallini.di.uniroma1.it/Graphs.html}}. All instances with up to a
  thousand vertices, apart from those with non-trivial initial colorings, were
  included. The experiments were run on a machine with an Intel i5-8600K
  processor at 3.6 GHz, and with 32 GB of RAM. Each instance was tested with a
  timeout of 5 seconds for the solver, and 3 minutes for the prover. No timeout
  was set for the verifier. A limit of 32 GB of RAM was set for each stage. All
  operations were single-threaded due to the nature of the prototype
  implementation.

  Out of 1107 attempted instances, the proof was successfully generated and
  verified for 469 of them, with the execution on the remaining ones being
  terminated early due to hitting any of the aforementioned limits. A breakdown
  of the results for different families of graphs is presented in Table
  \ref{table:results}. The results suggest that the approach is feasible for
  graphs in the order of hundreds of vertices, though this varies depending on
  the graph structure.
  \begin{table}[!h]
    \begin{center}
    {\small\begin{tabular}{ |l|c|c|c|l| }
  \hline
  \textbf{Family} & \textbf{Attempted} & \textbf{Verified} & \textbf{Largest verified $n$} & \textbf{Main failure} \\
  \hline
  ag & 12 & 5 & 105 & solver timeout \\
  \hline
  cfi & 41 & 24 & 680 & verifier memory \\
  \hline
  cmz & 37 & 8 & 288 & solver timeout \\
  \hline
  grid & 15 & 12 & 625 & verifier memory \\
  \hline
  grid-w & 15 & 15 & 1000 & -- \\
  \hline
  had & 64 & 22 & 512 & prover timeout \\
  \hline
  had-sw & 51 & 1 & 80 & solver timeout \\
  \hline
  k & 99 & 99 & 100 & -- \\
  \hline
  latin & 29 & 25 & 676 & verifier memory \\
  \hline
  latin-sw & 231 & 11 & 100 & prover timeout \\
  \hline
  lattice & 27 & 27 & 900 & -- \\
  \hline
  mz & 25 & 25 & 1000 & -- \\
  \hline
  mz-aug & 25 & 5 & 200 & solver timeout \\
  \hline
  mz-aug2 & 19 & 4 & 240 & solver timeout \\
  \hline
  paley & 53 & 53 & 461 & -- \\
  \hline
  pg & 12 & 4 & 62 & solver timeout \\
  \hline
  rnd-3-reg & 11 & 0 & -- & prover timeout \\
  \hline
  sat\_cfi\_dim & 70 & 44 & 700 & verifier memory \\
  \hline
  sts & 24 & 14 & 532 & prover timeout \\
  \hline
  sts-sw & 220 & 44 & 117 & prover timeout \\
  \hline
  triang & 27 & 27 & 435 & -- \\
  \hline
  \textbf{Total} & 1107 & 469 & & \\
  \hline
\end{tabular}
}
      \caption{Execution results per family, showing the number of attempted
      instances (i.e. excluding those with vertex counts over 1000), number of
      successfully verified instances (i.e. where the entire pipeline completed
      successfully), the largest verified instance and the most frequent
      failure reason.}
    \label{table:results}
    \end{center}
  \end{table}

  Among the three stages of the process, the prover and verifier dominate
  the execution time, with the solver time being negligible in comparison.
  Figure \ref{fig:time} shows the breakdown of the execution time for each
  of the three components. Note that all of the plots that follow consider
  only the 469 successfully verified instances.
  \begin{figure}[H]
  \begin{center}
    \input{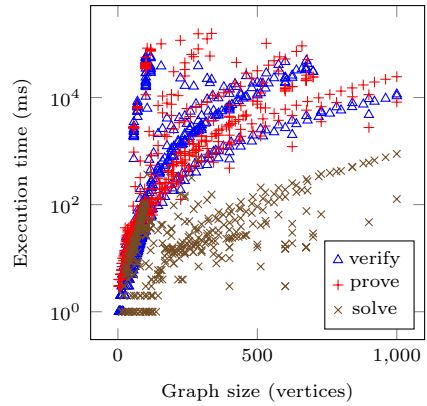}
  \end{center}
  \caption{Breakdown of execution time against graph size.}
  \label{fig:time}
  \end{figure}

  Figure \ref{fig:proof} shows the relation between the proof generation and
  verification times and the proof sizes. More specifically, the correlation
  between the proof generation time and the unpruned proof size is $r \approx
  0.99$. The correlation between the verification time and the final proof size
  is $r \approx 0.99$. This indicates that the primary bottleneck in both proof
  generation and verification is proof size. 
  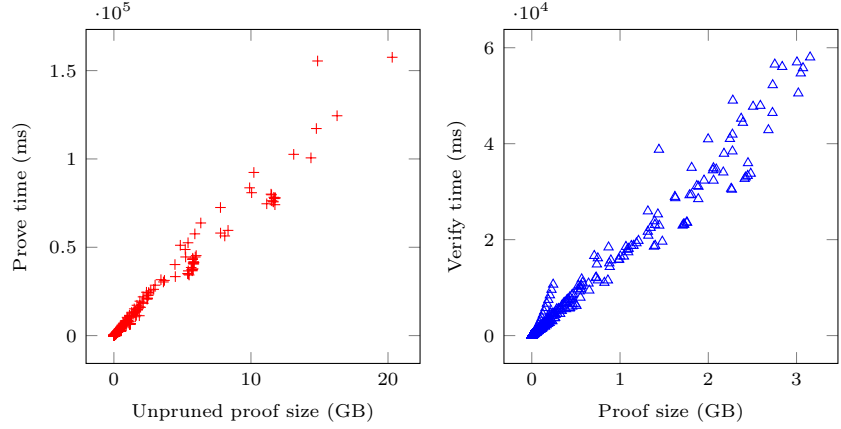
\begin{figure}[H]
  \begin{center}
    \begin{tikzpicture}
\begin{axis}[
    width=6cm,
    height=6cm,
    xlabel={Unpruned proof size (GB)},
    ylabel={Prove time (ms)},
    mark options=solid,
    unbounded coords=discard,
]
\addplot+[mark=+,red,only marks] coordinates {(0.0000,3) (0.0002,5) (0.0007,33) (0.0091,94) (9.9032,83660) (0.3282,2581) (0.3984,3406) (0.4305,3643) (0.5302,4418) (0.8107,6531) (0.9427,7412) (1.0227,8211) (1.2847,10992) (1.2896,10586) (1.3131,12294) (1.4574,12803) (2.1404,18248) (1.8683,16199) (3.0055,28539) (2.5041,22925) (3.4325,31651) (2.0534,18831) (4.4556,40225) (5.2049,48740) (5.9062,57593) (5.4099,52555) (7.7760,72510) (10.2117,92347) (4.8466,51138) (0.0387,367) (0.0897,895) (0.2071,2005) (0.4833,4619) (1.1381,10866) (2.6890,25287) (6.3399,63756) (14.8633,155451) (0.0000,3) (0.0002,6) (0.0003,8) (0.0033,46) (0.0107,92) (0.0235,197) (0.1136,951) (0.1011,830) (0.4386,3709) (0.5081,4339) (1.4075,12180) (1.8006,16019) (0.0000,4) (0.0001,10) (0.0002,7) (0.0015,48) (0.0028,34) (0.0066,64) (0.0273,242) (0.0143,125) (0.0686,620) (0.0511,480) (0.1988,1915) (0.1212,1213) (0.3798,3536) (0.2726,2757) (0.8619,8216) (0.0000,3) (0.0000,4) (0.0000,5) (0.0003,11) (0.0009,38) (0.0022,39) (0.0413,320) (0.0148,122) (0.3439,2320) (0.2973,1980) (1.0464,7098) (1.5643,10668) (1.1201,7884) (1.6397,11457) (14.7732,117212) (4.4815,33398) (0.1039,878) (7.7812,58042) (10.0497,80894) (16.2843,124446) (20.3081,157547) (0.7671,6440) (0.0877,653) (0.0000,3) (0.0000,4) (0.0000,4) (0.0000,3) (0.0000,4) (0.0000,3) (0.0000,3) (0.0000,4) (0.0000,5) (0.0000,5) (0.0000,6) (0.0000,6) (0.0000,4) (0.0001,7) (0.0001,10) (0.0001,8) (0.0001,5) (0.0001,9) (0.0001,6) (0.0001,11) (0.0002,12) (0.0002,13) (0.0002,14) (0.0002,6) (0.0002,6) (0.0003,18) (0.0003,19) (0.0003,21) (0.0004,22) (0.0004,24) (0.0005,26) (0.0005,27) (0.0005,21) (0.0006,31) (0.0006,24) (0.0007,16) (0.0008,15) (0.0008,28) (0.0009,21) (0.0009,16) (0.0010,16) (0.0011,24) (0.0012,29) (0.0012,32) (0.0013,32) (0.0014,23) (0.0015,41) (0.0016,28) (0.0017,28) (0.0018,21) (0.0019,28) (0.0020,28) (0.0021,27) (0.0022,24) (0.0023,31) (0.0025,35) (0.0026,28) (0.0027,28) (0.0029,29) (0.0030,31) (0.0032,35) (0.0033,38) (0.0035,37) (0.0036,37) (0.0038,45) (0.0040,39) (0.0042,42) (0.0043,46) (0.0045,44) (0.0047,46) (0.0049,47) (0.0051,51) (0.0053,52) (0.0056,53) (0.0058,56) (0.0060,58) (0.0062,59) (0.0065,61) (0.0067,63) (0.0070,66) (0.0072,68) (0.0075,77) (0.0078,73) (0.0081,75) (0.0083,78) (0.0086,79) (0.0089,82) (0.0092,85) (0.0096,87) (0.0099,90) (0.0102,95) (0.0105,96) (0.0109,98) (0.0112,100) (0.0116,104) (0.0119,108) (0.0123,111) (0.0127,114) (0.0131,116) (0.0000,3) (0.0000,3) (0.0001,6) (0.0002,5) (0.0006,9) (0.0008,12) (0.0027,32) (0.0044,42) (0.0095,75) (0.0094,82) (0.0407,311) (0.0245,212) (0.0686,564) (0.1151,901) (0.1917,1496) (0.1203,1021) (0.4796,3723) (0.2344,1973) (0.9152,7255) (0.8806,7059) (1.0462,8639) (0.7447,6399) (3.6899,31254) (1.9470,16016) (2.9024,26181) (5.7529,43803) (5.8035,40648) (5.8453,41448) (5.7442,43136) (5.7811,41096) (5.7505,43178) (5.8251,41205) (8.0890,56401) (5.8329,40588) (5.8189,41804) (5.9378,44499) (0.0001,4) (0.0002,6) (0.0004,8) (0.0008,42) (0.0016,36) (0.0029,47) (0.0049,64) (0.0079,78) (0.0122,124) (0.0181,187) (0.0261,248) (0.0367,348) (0.0504,483) (0.0681,658) (0.0903,873) (0.1181,1132) (0.1522,1461) (0.1938,1888) (0.2441,2338) (0.3043,2937) (0.3757,3758) (0.4601,4598) (0.5588,5814) (0.6739,6733) (0.8072,8106) (0.9607,10337) (1.1368,11339) (0.0011,13) (0.0049,64) (0.0123,105) (0.0238,195) (0.0400,341) (0.0618,535) (0.0896,801) (0.1241,1121) (0.1661,1502) (0.2161,1984) (0.2748,2543) (0.3429,3230) (0.4210,4004) (0.5097,5037) (0.6098,6026) (0.7218,7178) (0.8464,8473) (0.9844,9878) (1.1362,11461) (1.3026,13213) (1.4843,15135) (1.6818,17176) (1.8959,19492) (2.1272,21942) (2.3763,24549) (0.0015,22) (0.0157,162) (0.1014,837) (0.6147,5178) (3.6157,30371) (0.0226,197) (0.1474,1267) (0.8920,7486) (5.2330,44522) (0.0000,3) (0.0000,5) (0.0000,8) (0.0001,6) (0.0002,11) (0.0002,13) (0.0004,21) (0.0006,9) (0.0009,19) (0.0011,13) (0.0016,30) (0.0027,25) (0.0039,39) (0.0047,42) (0.0060,52) (0.0068,57) (0.0084,71) (0.0093,78) (0.0116,104) (0.0128,110) (0.0162,132) (0.0206,169) (0.0240,197) (0.0303,257) (0.0317,260) (0.0363,297) (0.0438,368) (0.0464,398) (0.0721,587) (0.0757,625) (0.0836,685) (0.1010,823) (0.1155,946) (0.1258,1026) (0.1313,1085) (0.1446,1236) (0.1482,1217) (0.1805,1500) (0.1872,1559) (0.2244,1881) (0.2486,2103) (0.2572,2115) (0.2774,2435) (0.3027,2576) (0.3424,2836) (0.3639,3045) (0.3748,3117) (0.3975,3303) (0.4330,3591) (0.4708,3882) (0.5238,4508) (0.5524,4557) (0.5663,4729) (0.0000,4) (0.0002,6) (0.0007,12) (0.0220,218) (0.0003,7) (0.0003,16) (0.0022,21) (0.0022,20) (0.0072,66) (0.0068,61) (0.0099,86) (0.0089,75) (0.0154,126) (0.0163,134) (0.0303,250) (0.0296,242) (0.0503,432) (0.0495,423) (0.0632,531) (0.0644,543) (0.0777,631) (0.0773,625) (0.1207,980) (0.1138,908) (0.1666,1355) (0.1673,1365) (0.2330,1918) (0.2320,1907) (0.2262,1902) (0.1955,1604) (0.5215,4488) (0.5006,4298) (1.0005,8491) (0.8883,7542) (1.7148,14571) (1.7294,14825) (1.7291,14743) (1.7610,15166) (1.7320,14645) (1.5765,13752) (1.4973,13158) (2.4317,20735) (2.4542,21206) (2.4575,21051) (2.4886,22805) (2.4198,20756) (2.5439,24239) (2.5320,24255) (0.0000,4) (0.0000,4) (0.0010,37) (0.0023,35) (0.1705,958) (0.0126,79) (1.8563,11269) (0.0843,545) (14.3786,100609) (0.2734,1873) (0.9046,6239) (8.3286,59558) (6.0052,45285) (13.1199,102604) (0.5101,2779) (0.5120,2825) (0.5065,2744) (0.5053,2740) (0.5016,2749) (0.5033,2773) (0.4927,2676) (0.5045,2748) (0.5089,2813) (0.4851,2669) (0.5077,2861) (1.1816,6666) (1.1446,6504) (1.1627,6649) (1.2014,6784) (1.2072,6790) (1.1691,6544) (1.1998,6577) (1.1607,6323) (1.1716,6506) (1.1983,6646) (1.1490,6341) (5.7796,37815) (5.6773,37374) (5.6404,38546) (5.4346,34788) (5.4033,35120) (5.3850,35097) (5.7024,36996) (5.7525,37363) (5.7204,36689) (5.4899,34319) (5.4268,36640) (11.7483,74101) (11.7629,77862) (11.1520,74625) (11.6373,75692) (11.5106,76318) (11.6718,78282) (11.4636,80072) (11.5023,79709) (11.6114,76077) (11.7049,77562) (11.7712,78006) (0.0000,3) (0.0000,5) (0.0000,6) (0.0001,10) (0.0002,14) (0.0004,8) (0.0007,18) (0.0012,45) (0.0018,19) (0.0027,38) (0.0040,39) (0.0057,55) (0.0078,72) (0.0106,103) (0.0141,131) (0.0184,170) (0.0238,215) (0.0303,269) (0.0382,340) (0.0477,431) (0.0589,534) (0.0721,655) (0.0875,792) (0.1055,945) (0.1263,1122) (0.1503,1362) (0.1778,1598)};
\end{axis}
\end{tikzpicture}
    \begin{tikzpicture}
\begin{axis}[
    width=6cm,
    height=6cm,
    xlabel={Proof size (GB)},
    ylabel={Verify time (ms)},
    mark options=solid,
    unbounded coords=discard,
]
\addplot+[mark=triangle,only marks] coordinates {(0.0000,1) (0.0001,5) (0.0005,16) (0.0022,31) (0.5067,6232) (0.1500,2485) (0.2321,4157) (0.2540,4523) (0.3220,5504) (0.4567,7552) (0.2903,4582) (0.5428,9257) (0.6432,10883) (0.5568,9777) (0.8962,15708) (0.8836,14341) (1.0824,17913) (0.7235,12017) (1.4287,25303) (0.8781,15095) (1.1862,19285) (1.0977,18120) (1.0670,18519) (1.3952,23097) (1.8891,31065) (2.7274,46442) (2.0596,32330) (2.1698,34038) (2.6812,42854) (0.0183,281) (0.0314,473) (0.0495,732) (0.0737,1109) (0.1050,1555) (0.1448,2098) (0.1948,2824) (0.2582,3642) (0.0000,0) (0.0002,4) (0.0003,5) (0.0030,45) (0.0100,137) (0.0224,303) (0.1100,1447) (0.0970,1291) (0.4285,5685) (0.4934,6574) (1.3845,18575) (1.7601,23569) (0.0000,1) (0.0001,2) (0.0001,3) (0.0011,20) (0.0021,32) (0.0057,83) (0.0225,323) (0.0114,188) (0.0603,978) (0.0385,706) (0.1702,2571) (0.0970,2018) (0.3399,5193) (0.2068,4762) (0.7479,11537) (0.0000,0) (0.0000,1) (0.0000,1) (0.0003,23) (0.0008,21) (0.0018,36) (0.0084,129) (0.0069,106) (0.0466,693) (0.0527,915) (0.1252,2106) (0.1566,2607) (0.1983,4091) (0.1701,2947) (1.4445,22957) (0.7080,16633) (0.0959,1541) (0.7482,16147) (1.4408,38769) (1.3466,22486) (1.3147,21685) (0.7302,12120) (0.0135,197) (0.0000,0) (0.0000,1) (0.0000,1) (0.0000,1) (0.0000,1) (0.0000,1) (0.0000,1) (0.0000,2) (0.0000,2) (0.0000,2) (0.0000,2) (0.0000,2) (0.0000,0) (0.0000,2) (0.0000,0) (0.0000,1) (0.0000,3) (0.0000,1) (0.0000,3) (0.0000,3) (0.0000,4) (0.0000,4) (0.0000,4) (0.0000,2) (0.0000,2) (0.0000,6) (0.0000,6) (0.0000,6) (0.0000,5) (0.0000,5) (0.0001,4) (0.0001,6) (0.0001,3) (0.0001,5) (0.0001,3) (0.0001,4) (0.0001,4) (0.0001,8) (0.0001,3) (0.0001,4) (0.0001,4) (0.0001,7) (0.0001,5) (0.0001,6) (0.0001,5) (0.0001,6) (0.0001,7) (0.0002,5) (0.0002,7) (0.0002,5) (0.0002,6) (0.0002,5) (0.0002,6) (0.0002,6) (0.0002,9) (0.0002,7) (0.0002,7) (0.0002,7) (0.0003,8) (0.0003,8) (0.0003,9) (0.0003,9) (0.0003,10) (0.0003,10) (0.0003,11) (0.0003,11) (0.0004,12) (0.0004,12) (0.0004,13) (0.0004,13) (0.0004,14) (0.0004,14) (0.0004,15) (0.0005,15) (0.0005,16) (0.0005,17) (0.0005,18) (0.0005,18) (0.0005,19) (0.0006,20) (0.0006,21) (0.0006,21) (0.0006,23) (0.0006,23) (0.0007,24) (0.0007,25) (0.0007,26) (0.0007,27) (0.0007,28) (0.0008,29) (0.0008,30) (0.0008,31) (0.0008,32) (0.0009,33) (0.0009,34) (0.0009,36) (0.0009,37) (0.0010,38) (0.0010,39) (0.0000,0) (0.0000,1) (0.0000,5) (0.0001,8) (0.0005,9) (0.0008,17) (0.0022,46) (0.0035,68) (0.0072,141) (0.0090,183) (0.0265,502) (0.0238,483) (0.0523,1025) (0.0742,1435) (0.1311,2590) (0.1182,2394) (0.2817,5512) (0.2311,4651) (0.5654,11768) (0.5751,11058) (0.8699,18415) (0.7375,14865) (1.9982,40952) (1.3156,25928) (2.2779,49015) (3.0048,56982) (2.2735,41932) (2.2449,41022) (2.7297,52232) (2.5075,47739) (2.8377,56028) (2.3728,45240) (3.0201,50514) (2.3921,44384) (2.5901,47923) (2.7534,56534) (0.0000,1) (0.0001,2) (0.0002,4) (0.0005,16) (0.0010,15) (0.0018,24) (0.0029,39) (0.0047,62) (0.0072,95) (0.0106,140) (0.0153,203) (0.0215,287) (0.0295,397) (0.0397,537) (0.0526,722) (0.0686,997) (0.0883,1219) (0.1122,1556) (0.1411,1968) (0.1757,2459) (0.2166,3045) (0.2650,3725) (0.3215,4557) (0.3873,5544) (0.4635,6629) (0.5512,7937) (0.6517,9427) (0.0004,7) (0.0016,24) (0.0036,58) (0.0064,109) (0.0099,182) (0.0142,278) (0.0192,401) (0.0251,553) (0.0317,739) (0.0391,960) (0.0472,1223) (0.0561,1522) (0.0658,1894) (0.0763,2266) (0.0875,2708) (0.0995,3209) (0.1123,3755) (0.1258,4390) (0.1401,5049) (0.1552,5779) (0.1711,6683) (0.1877,7469) (0.2051,8431) (0.2232,9577) (0.2422,10656) (0.0009,19) (0.0060,83) (0.0193,259) (0.0457,652) (0.0939,1209) (0.0097,132) (0.0318,421) (0.0750,975) (0.1492,1911) (0.0000,1) (0.0000,2) (0.0000,3) (0.0000,5) (0.0001,13) (0.0002,18) (0.0004,24) (0.0005,11) (0.0009,31) (0.0011,19) (0.0016,52) (0.0026,48) (0.0037,65) (0.0046,82) (0.0059,108) (0.0066,118) (0.0082,149) (0.0091,166) (0.0113,207) (0.0125,231) (0.0159,282) (0.0203,372) (0.0237,449) (0.0296,533) (0.0313,567) (0.0358,638) (0.0433,800) (0.0458,840) (0.0714,1318) (0.0750,1355) (0.0828,1508) (0.1001,1841) (0.1145,2082) (0.1247,2296) (0.1302,2416) (0.1426,2620) (0.1470,2823) (0.1791,3293) (0.1858,3396) (0.2229,4102) (0.2470,4539) (0.2555,4720) (0.2745,5085) (0.3008,5579) (0.3404,6254) (0.3618,6756) (0.3726,6875) (0.3952,7341) (0.4306,8005) (0.4683,8670) (0.5210,9630) (0.5496,10214) (0.5634,10407) (0.0000,1) (0.0001,4) (0.0004,8) (0.0032,38) (0.0003,7) (0.0003,25) (0.0021,35) (0.0022,32) (0.0071,98) (0.0066,94) (0.0087,120) (0.0077,108) (0.0152,208) (0.0161,220) (0.0299,407) (0.0292,396) (0.0497,671) (0.0490,661) (0.0583,785) (0.0596,797) (0.0770,1033) (0.0766,1028) (0.1198,1594) (0.1128,1512) (0.1654,2217) (0.1661,2205) (0.2315,3063) (0.2304,3046) (0.2024,2673) (0.1845,2482) (0.4342,5794) (0.4582,6083) (0.8625,11526) (0.8226,11042) (1.7091,22965) (1.7234,23161) (1.7230,23146) (1.7547,23589) (1.7259,23203) (1.4792,19615) (1.3980,18747) (2.4241,33005) (2.4467,33325) (2.4499,35974) (2.4809,33772) (2.4124,32704) (2.2667,30448) (2.2595,30634) (0.0000,0) (0.0000,0) (0.0004,13) (0.0009,13) (0.0452,715) (0.0049,78) (0.5279,10451) (0.0272,443) (2.2738,38423) (0.0769,1246) (0.2043,3235) (1.1528,18746) (0.9906,15876) (1.8863,28479) (0.1309,2073) (0.1712,2917) (0.1187,1834) (0.1228,1930) (0.1425,2423) (0.1510,2589) (0.1215,1953) (0.1212,1913) (0.1233,1932) (0.1477,2297) (0.1922,3498) (0.4048,7497) (0.4135,7566) (0.3316,5687) (0.3908,6792) (0.3946,6892) (0.3108,5470) (0.2708,4367) (0.3258,5527) (0.3581,6565) (0.3975,6902) (0.4536,8419) (1.1278,18716) (1.6217,28946) (1.6262,28676) (1.0437,16506) (1.3202,20860) (1.0929,17399) (1.0031,16533) (1.2089,19787) (0.9836,15793) (1.3917,23862) (1.8114,35003) (2.0905,34654) (2.0528,34483) (2.1777,37936) (1.7879,29291) (3.1538,58013) (1.8702,31204) (3.0478,54683) (3.0733,55771) (1.8083,29387) (1.9517,32342) (2.0584,34927) (0.0000,1) (0.0000,2) (0.0000,4) (0.0001,7) (0.0002,14) (0.0003,6) (0.0006,26) (0.0009,17) (0.0014,24) (0.0021,28) (0.0031,45) (0.0043,57) (0.0059,79) (0.0080,107) (0.0105,141) (0.0137,183) (0.0176,235) (0.0224,301) (0.0280,378) (0.0348,466) (0.0428,579) (0.0522,702) (0.0632,851) (0.0759,1025) (0.0905,1222) (0.1073,1463) (0.1265,1732)};
\end{axis}
\end{tikzpicture}
  \end{center}
  \caption{Breakdown of execution time against proof size. The left plot
    compares prover execution time against the unpruned proof size, while the right
    plot compares verifier execution time against the final proof size.}
  \label{fig:proof}
  \end{figure}

  Before committing the generated proof, the prover prunes any proof steps
  which are not necessary for deriving the final polynomial. Figure
  \ref{fig:proofprune} shows the relation between the unpruned and pruned proof
  sizes, as well as the time taken for the pruning operation. Considering that
  the pruning operation takes a small portion of the total proof
  generation time, and that it often reduced the generated proof size by more
  than 50\%, it appears to be a worthwhile addition to the prover
  implementation. Furthermore, this indicates that there is room for
  improvement in the proof construction implementation itself, as a significant
  portion of the constructed proof is redundant in many cases.
  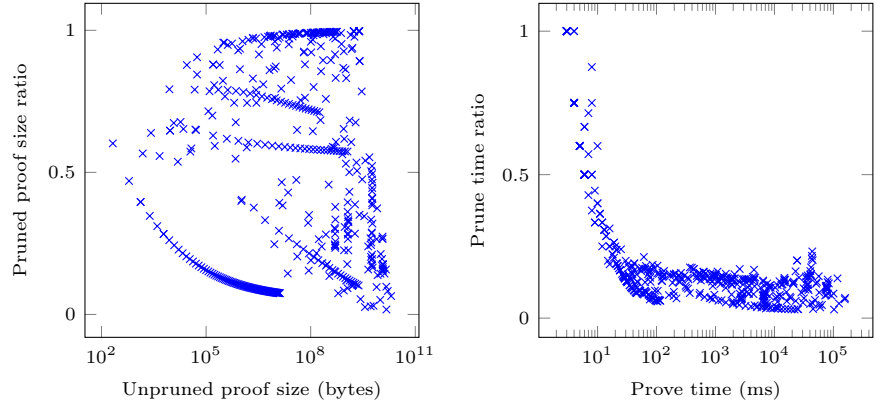
\begin{figure}[tbp]
  \begin{center}
    \begin{tikzpicture}
\begin{axis}[
    width=6cm,
    height=6cm,
    xlabel={Unpruned proof size (bytes)},
    ylabel={Pruned proof size ratio},
    xmode=log,
    mark options=solid,
    unbounded coords=discard,
]
\addplot+[mark=x,only marks] coordinates {(22333,0.6528) (151510,0.7000) (708951,0.6763) (9068306,0.2472) (9903193391,0.0512) (328176592,0.4570) (398361230,0.5826) (430490936,0.5900) (530208878,0.6072) (810703672,0.5633) (942694764,0.3079) (1022698616,0.5308) (1284678784,0.5007) (1289644560,0.4317) (1313079118,0.6825) (1457357616,0.6063) (2140432380,0.5057) (1868273372,0.3872) (3005539216,0.4754) (2504057946,0.3507) (3432527484,0.3456) (2053396828,0.5346) (4455552436,0.2395) (5204868406,0.2681) (5906211922,0.3198) (5409863990,0.5041) (7775953978,0.2649) (10211680526,0.2125) (4846564200,0.5532) (38741331,0.4735) (89654897,0.3500) (207128547,0.2391) (483346153,0.1525) (1138139299,0.0923) (2688950161,0.0538) (6339896595,0.0307) (14863258985,0.0174) (9246,0.6469) (245440,0.8780) (304511,0.8225) (3277878,0.9151) (10668571,0.9390) (23515428,0.9505) (113634521,0.9676) (101107536,0.9596) (438639527,0.9768) (508140575,0.9711) (1407489929,0.9837) (1800627372,0.9775) (9246,0.6469) (125236,0.7855) (170859,0.8071) (1493261,0.7446) (2780756,0.7448) (6603135,0.8595) (27320878,0.8251) (14349645,0.7947) (68557206,0.8801) (51063399,0.7539) (198814121,0.8558) (121232712,0.8004) (379766352,0.8950) (272610152,0.7587) (861899320,0.8678) (1520,0.5664) (8866,0.7924) (49230,0.6509) (347690,0.7674) (935665,0.8629) (2194433,0.8304) (41311588,0.2039) (14771426,0.4653) (343869199,0.1354) (297272881,0.1774) (1046431254,0.1196) (1564348564,0.1001) (1120117284,0.1770) (1639725303,0.1038) (14773219268,0.0978) (4481456655,0.1580) (103938801,0.9224) (7781211188,0.0962) (10049713794,0.1434) (16284342723,0.0827) (20308148538,0.0647) (767101603,0.9518) (87691896,0.1543) (211,0.6019) (615,0.4699) (1326,0.3959) (2420,0.3467) (3973,0.3108) (6061,0.2833) (8760,0.2613) (12146,0.2433) (16295,0.2282) (21283,0.2154) (27186,0.2044) (34080,0.1948) (42041,0.1864) (51145,0.1789) (61468,0.1722) (73086,0.1662) (86075,0.1608) (100511,0.1559) (116470,0.1514) (134028,0.1472) (153261,0.1435) (174245,0.1400) (197056,0.1367) (221770,0.1337) (248463,0.1309) (277211,0.1283) (308090,0.1259) (341176,0.1236) (376545,0.1215) (414273,0.1195) (454436,0.1175) (497110,0.1157) (542371,0.1140) (590295,0.1124) (640958,0.1109) (694436,0.1095) (750805,0.1081) (810141,0.1068) (872520,0.1055) (938018,0.1043) (1006711,0.1032) (1078675,0.1021) (1153986,0.1010) (1232720,0.1000) (1314953,0.0990) (1400761,0.0981) (1490220,0.0972) (1583406,0.0964) (1680395,0.0955) (1781263,0.0948) (1886086,0.0940) (1994940,0.0933) (2107901,0.0925) (2225045,0.0919) (2346448,0.0912) (2472186,0.0906) (2602335,0.0899) (2736971,0.0893) (2876170,0.0888) (3020008,0.0882) (3168561,0.0876) (3321905,0.0871) (3480116,0.0866) (3643270,0.0861) (3811443,0.0856) (3984711,0.0851) (4163150,0.0847) (4346836,0.0842) (4535845,0.0838) (4730253,0.0834) (4930136,0.0830) (5135570,0.0826) (5346631,0.0822) (5563395,0.0818) (5785938,0.0815) (6014336,0.0811) (6248665,0.0807) (6489001,0.0804) (6735420,0.0801) (6987998,0.0797) (7246811,0.0794) (7511935,0.0791) (7783446,0.0788) (8061420,0.0785) (8345933,0.0782) (8637061,0.0779) (8934880,0.0777) (9239466,0.0774) (9550895,0.0771) (9869243,0.0769) (10194586,0.0766) (10527000,0.0764) (10866561,0.0761) (11213345,0.0759) (11567428,0.0756) (11928886,0.0754) (12297795,0.0752) (12674231,0.0750) (13058270,0.0748) (1326,0.3959) (15435,0.5375) (55232,0.8443) (159511,0.8795) (638744,0.7450) (833539,0.9292) (2727614,0.7898) (4379067,0.8098) (9482474,0.7622) (9360468,0.9609) (40740415,0.6494) (24544292,0.9709) (68578340,0.7629) (115108708,0.6448) (191709938,0.6839) (120289136,0.9824) (479553174,0.5875) (234432690,0.9858) (915160292,0.6179) (880583739,0.6531) (1046198065,0.8315) (744723512,0.9903) (3689934645,0.5415) (1947022038,0.6757) (2902374583,0.7848) (5752868405,0.5223) (5803515612,0.3918) (5845346822,0.3841) (5744208997,0.4752) (5781138374,0.4337) (5750478294,0.4935) (5825124007,0.4073) (8088997469,0.3734) (5832897522,0.4101) (5818867074,0.4451) (5937833275,0.4637) (51618,0.6491) (158560,0.6296) (393474,0.6184) (845648,0.6109) (1638196,0.6053) (2933448,0.6009) (4938340,0.5973) (7909804,0.5943) (12160158,0.5917) (18062496,0.5896) (26056078,0.5876) (36651720,0.5859) (50437184,0.5844) (68082568,0.5831) (90345696,0.5819) (118077508,0.5808) (152227450,0.5798) (193848864,0.5789) (244104378,0.5781) (304271296,0.5773) (375746988,0.5766) (460054280,0.5759) (558846844,0.5753) (673914588,0.5747) (807189046,0.5742) (960748768,0.5737) (1136824710,0.5732) (1051197,0.4037) (4923998,0.3320) (12278883,0.2945) (23768652,0.2679) (40046105,0.2470) (61764042,0.2296) (89575263,0.2149) (124132568,0.2021) (166088757,0.1908) (216096630,0.1808) (274808987,0.1718) (342878628,0.1637) (420958353,0.1563) (509700962,0.1496) (609759255,0.1435) (721786032,0.1379) (846434093,0.1326) (984356238,0.1278) (1136205267,0.1233) (1302633980,0.1192) (1484295177,0.1152) (1681841658,0.1116) (1895926223,0.1082) (2127201672,0.1049) (2376320805,0.1019) (1464986,0.5881) (15650430,0.3836) (101407000,0.1905) (614665234,0.0743) (3615700812,0.0260) (22557991,0.4292) (147410635,0.2156) (891951519,0.0841) (5232954019,0.0285) (2631,0.6389) (11827,0.6738) (28082,0.8779) (55139,0.9052) (159511,0.8795) (215345,0.9325) (418851,0.9485) (554504,0.9535) (942255,0.9354) (1109208,0.9613) (1646923,0.9668) (2703358,0.9714) (3872459,0.9514) (4707587,0.9757) (6009795,0.9775) (6759104,0.9783) (8401159,0.9797) (9327591,0.9805) (11647800,0.9707) (12832553,0.9762) (16174838,0.9835) (20588418,0.9851) (24013790,0.9858) (30259482,0.9782) (31713537,0.9866) (36263056,0.9873) (43764354,0.9883) (46378047,0.9883) (72097354,0.9899) (75727617,0.9900) (83608114,0.9904) (100990028,0.9910) (115471568,0.9914) (125753669,0.9917) (131250691,0.9918) (144590488,0.9866) (148178000,0.9920) (180493443,0.9925) (187220128,0.9926) (224409187,0.9931) (248629978,0.9933) (257215303,0.9934) (277403652,0.9894) (302707656,0.9938) (342446430,0.9940) (363902605,0.9941) (374800933,0.9941) (397505417,0.9943) (432961430,0.9944) (470801332,0.9946) (523766986,0.9948) (552413805,0.9949) (566330356,0.9949) (37794,0.5843) (177349,0.5933) (700110,0.5481) (21990516,0.1440) (332317,0.9567) (311658,0.9566) (2199918,0.9737) (2215315,0.9716) (7224984,0.9816) (6774120,0.9808) (9921777,0.8727) (8895327,0.8653) (15380988,0.9850) (16334845,0.9855) (30277519,0.9877) (29594166,0.9880) (50260882,0.9891) (49501962,0.9898) (63226742,0.9222) (64418888,0.9252) (77651503,0.9911) (77304132,0.9911) (120738633,0.9919) (113760816,0.9920) (166623495,0.9928) (167258298,0.9928) (233045367,0.9934) (231969376,0.9932) (226246361,0.8946) (195518166,0.9438) (521517029,0.8326) (500561557,0.9153) (1000499559,0.8621) (888328741,0.9260) (1714849674,0.9966) (1729395527,0.9966) (1729072279,0.9965) (1761007343,0.9964) (1731951412,0.9965) (1576479915,0.9383) (1497348693,0.9336) (2431651215,0.9969) (2454225607,0.9969) (2457492786,0.9969) (2488554034,0.9969) (2419755647,0.9970) (2543894947,0.8910) (2532002253,0.8924) (6061,0.2833) (35376,0.5730) (1022850,0.3992) (2332782,0.3761) (170488216,0.2649) (12595080,0.3929) (1856273893,0.2844) (84333923,0.3223) (14378558807,0.1581) (273393268,0.2814) (904617681,0.2258) (8328557236,0.1384) (6005197027,0.1650) (13119934030,0.1438) (510089209,0.2567) (511998072,0.3343) (506476278,0.2343) (505297525,0.2430) (501613401,0.2841) (503270472,0.3000) (492672733,0.2467) (504544519,0.2402) (508911142,0.2423) (485099188,0.3045) (507685064,0.3786) (1181598755,0.3426) (1144606891,0.3613) (1162718263,0.2852) (1201390360,0.3253) (1207231811,0.3269) (1169063767,0.2659) (1199803047,0.2257) (1160694462,0.2807) (1171621157,0.3056) (1198258448,0.3317) (1148989486,0.3947) (5779648542,0.1951) (5677289160,0.2856) (5640361040,0.2883) (5434575867,0.1920) (5403266975,0.2443) (5385009891,0.2029) (5702373116,0.1759) (5752455904,0.2102) (5720408313,0.1719) (5489861202,0.2535) (5426818843,0.3338) (11748268835,0.1779) (11762946134,0.1745) (11151966583,0.1953) (11637294422,0.1536) (11510571606,0.2740) (11671777027,0.1602) (11463586967,0.2659) (11502342765,0.2672) (11611366309,0.1557) (11704941342,0.1667) (11771163726,0.1749) (4648,0.5577) (18532,0.6775) (47031,0.7912) (116240,0.7151) (224678,0.7910) (417035,0.7901) (720228,0.7868) (1175209,0.7824) (1831630,0.7775) (2748641,0.7724) (3995688,0.7674) (5653311,0.7625) (7813942,0.7578) (10582703,0.7532) (14078204,0.7489) (18433341,0.7449) (23796094,0.7410) (30330325,0.7373) (38216576,0.7338) (47652867,0.7305) (58855494,0.7274) (72059827,0.7244) (87521108,0.7216) (105515249,0.7189) (126339630,0.7163) (150313897,0.7139) (177780760,0.7115)};
\end{axis}
\end{tikzpicture}
    \begin{tikzpicture}
\begin{axis}[
    width=6cm,
    height=6cm,
    xlabel={Prove time (ms)},
    ylabel={Prune time ratio},
    xmode=log,
    mark options=solid,
    unbounded coords=discard,
]
\addplot+[mark=x,only marks] coordinates {(3,1.0000) (5,0.6000) (33,0.1515) (94,0.0957) (83660,0.1209) (2581,0.0852) (3406,0.0907) (3643,0.0840) (4418,0.0835) (6531,0.0896) (7412,0.0529) (8211,0.0770) (10992,0.0758) (10586,0.0691) (12294,0.0813) (12803,0.0759) (18248,0.0743) (16199,0.0593) (28539,0.0894) (22925,0.0600) (31651,0.0692) (18831,0.0698) (40225,0.0494) (48740,0.0632) (57593,0.0893) (52555,0.1258) (72510,0.0787) (92347,0.0785) (51138,0.1197) (367,0.0872) (895,0.0659) (2005,0.0519) (4619,0.0446) (10866,0.0360) (25287,0.0312) (63756,0.0457) (155451,0.0657) (3,1.0000) (6,0.5000) (8,0.5000) (46,0.1522) (92,0.1739) (197,0.1523) (951,0.1504) (830,0.1349) (3709,0.1316) (4339,0.1288) (12180,0.1287) (16019,0.1220) (4,1.0000) (10,0.4000) (7,0.4286) (48,0.1250) (34,0.1765) (64,0.1719) (242,0.1281) (125,0.1440) (620,0.1435) (480,0.1375) (1915,0.1154) (1213,0.1179) (3536,0.1075) (2757,0.0932) (8216,0.1059) (3,1.0000) (4,0.7500) (5,0.6000) (11,0.3636) (38,0.1316) (39,0.1538) (320,0.1062) (122,0.1311) (2320,0.0539) (1980,0.0581) (7098,0.0425) (10668,0.0393) (7884,0.0460) (11457,0.0363) (117212,0.1385) (33398,0.0991) (878,0.1481) (58042,0.0470) (80894,0.0854) (124446,0.0509) (157547,0.0705) (6440,0.1311) (653,0.0735) (3,1.0000) (4,0.7500) (4,0.7500) (3,1.0000) (4,0.7500) (3,1.0000) (3,1.0000) (4,0.7500) (5,0.6000) (5,0.6000) (6,0.6667) (6,0.5000) (4,0.7500) (7,0.5714) (10,0.6000) (8,0.7500) (5,0.6000) (9,0.3333) (6,0.5000) (11,0.3636) (12,0.2500) (13,0.3077) (14,0.2857) (6,0.5000) (6,0.5000) (18,0.2222) (19,0.2105) (21,0.1905) (22,0.1818) (24,0.1667) (26,0.1538) (27,0.1481) (21,0.1905) (31,0.1290) (24,0.1667) (16,0.2500) (15,0.2000) (28,0.1429) (21,0.1905) (16,0.2500) (16,0.1875) (24,0.1667) (29,0.1379) (32,0.1250) (32,0.1250) (23,0.1739) (41,0.0976) (28,0.1429) (28,0.1429) (21,0.1905) (28,0.1429) (28,0.1429) (27,0.1481) (24,0.1667) (31,0.1290) (35,0.1143) (28,0.1429) (28,0.1429) (29,0.1379) (31,0.1290) (35,0.1143) (38,0.1053) (37,0.1081) (37,0.1081) (45,0.0889) (39,0.1026) (42,0.0952) (46,0.0870) (44,0.0909) (46,0.0870) (47,0.1064) (51,0.0980) (52,0.0962) (53,0.0943) (56,0.0893) (58,0.0862) (59,0.0847) (61,0.0820) (63,0.0952) (66,0.0909) (68,0.0882) (77,0.0649) (73,0.0822) (75,0.0800) (78,0.0769) (79,0.0759) (82,0.0732) (85,0.0706) (87,0.0690) (90,0.0667) (95,0.0632) (96,0.0625) (98,0.0612) (100,0.0600) (104,0.0577) (108,0.0648) (111,0.0631) (114,0.0614) (116,0.0603) (3,1.0000) (3,1.0000) (6,0.6667) (5,0.6000) (9,0.3333) (12,0.3333) (32,0.1875) (42,0.1667) (75,0.1467) (82,0.1707) (311,0.1125) (212,0.1557) (564,0.1312) (901,0.1121) (1496,0.1056) (1021,0.1450) (3723,0.0854) (1973,0.1358) (7255,0.0838) (7059,0.0902) (8639,0.1060) (6399,0.1277) (31254,0.1040) (16016,0.0910) (26181,0.1414) (43803,0.2324) (40648,0.1504) (41448,0.1469) (43136,0.2171) (41096,0.1526) (43178,0.2068) (41205,0.1507) (56401,0.1459) (40588,0.1639) (41804,0.1769) (44499,0.2018) (4,0.7500) (6,0.5000) (8,0.3750) (42,0.1905) (36,0.1389) (47,0.1277) (64,0.1250) (78,0.1410) (124,0.1290) (187,0.1230) (248,0.1250) (348,0.1351) (483,0.1470) (658,0.1368) (873,0.1271) (1132,0.1263) (1461,0.1232) (1888,0.1186) (2338,0.1180) (2937,0.1171) (3758,0.1128) (4598,0.1111) (5814,0.1060) (6733,0.1132) (8106,0.1126) (10337,0.1063) (11339,0.1147) (13,0.3077) (64,0.0938) (105,0.0952) (195,0.0769) (341,0.1026) (535,0.0804) (801,0.0699) (1121,0.0607) (1502,0.0546) (1984,0.0489) (2543,0.0460) (3230,0.0427) (4004,0.0407) (5037,0.0381) (6026,0.0368) (7178,0.0337) (8473,0.0348) (9878,0.0345) (11461,0.0330) (13213,0.0321) (15135,0.0316) (17176,0.0311) (19492,0.0303) (21942,0.0301) (24549,0.0321) (22,0.1818) (162,0.0802) (837,0.0717) (5178,0.0413) (30371,0.0541) (197,0.0914) (1267,0.0624) (7486,0.0381) (44522,0.0515) (3,1.0000) (5,0.6000) (8,0.8750) (6,0.5000) (11,0.3636) (13,0.3077) (21,0.2381) (9,0.4444) (19,0.2632) (13,0.3077) (30,0.2000) (25,0.2400) (39,0.1795) (42,0.1905) (52,0.1731) (57,0.1754) (71,0.1690) (78,0.1667) (104,0.1827) (110,0.1545) (132,0.1515) (169,0.1538) (197,0.1523) (257,0.1401) (260,0.1577) (297,0.1481) (368,0.1766) (398,0.1759) (587,0.1465) (625,0.1600) (685,0.1562) (823,0.1555) (946,0.1512) (1026,0.1501) (1085,0.1622) (1236,0.1416) (1217,0.1454) (1500,0.1547) (1559,0.1430) (1881,0.1388) (2103,0.1526) (2115,0.1390) (2435,0.1298) (2576,0.1712) (2836,0.1365) (3045,0.1353) (3117,0.1347) (3303,0.1368) (3591,0.1351) (3882,0.1342) (4508,0.1686) (4557,0.1339) (4729,0.1413) (4,0.7500) (6,0.5000) (12,0.3333) (218,0.0688) (7,0.7143) (16,0.3125) (21,0.2381) (20,0.2500) (66,0.1818) (61,0.2131) (86,0.1860) (75,0.1600) (126,0.1587) (134,0.1642) (250,0.1440) (242,0.1488) (432,0.1713) (423,0.1702) (531,0.1582) (543,0.1565) (631,0.1616) (625,0.1616) (980,0.1500) (908,0.1553) (1355,0.1446) (1365,0.1429) (1918,0.1382) (1907,0.1311) (1902,0.1325) (1604,0.1378) (4488,0.1404) (4298,0.1438) (8491,0.1149) (7542,0.1208) (14571,0.1250) (14825,0.1350) (14743,0.1404) (15166,0.1333) (14645,0.1322) (13752,0.1175) (13158,0.1162) (20735,0.1229) (21206,0.1276) (21051,0.1186) (22805,0.1631) (20756,0.1215) (24239,0.2005) (24255,0.2013) (4,1.0000) (4,0.7500) (37,0.1351) (35,0.1429) (958,0.0814) (79,0.1266) (11269,0.0566) (545,0.0881) (100609,0.0857) (1873,0.0822) (6239,0.0410) (59558,0.0375) (45285,0.0319) (102604,0.0300) (2779,0.0677) (2825,0.0782) (2744,0.0649) (2740,0.0661) (2749,0.0713) (2773,0.0685) (2676,0.0617) (2748,0.0659) (2813,0.0711) (2669,0.0727) (2861,0.0905) (6666,0.0774) (6504,0.0801) (6649,0.0826) (6784,0.0786) (6790,0.0754) (6544,0.0604) (6577,0.0537) (6323,0.0783) (6506,0.0666) (6646,0.0713) (6341,0.0793) (37815,0.0975) (37374,0.1406) (38546,0.1437) (34788,0.1255) (35120,0.1369) (35097,0.1156) (36996,0.1034) (37363,0.1173) (36689,0.1150) (34319,0.1456) (36640,0.1487) (74101,0.1118) (77862,0.1089) (74625,0.1076) (75692,0.0809) (76318,0.1132) (78282,0.0982) (80072,0.1342) (79709,0.1252) (76077,0.1050) (77562,0.1018) (78006,0.1019) (3,1.0000) (5,0.6000) (6,0.5000) (10,0.4000) (14,0.2857) (8,0.5000) (18,0.2222) (45,0.1111) (19,0.2632) (38,0.1579) (39,0.1795) (55,0.1636) (72,0.1528) (103,0.1553) (131,0.1450) (170,0.1471) (215,0.1395) (269,0.1338) (340,0.1412) (431,0.1531) (534,0.1479) (655,0.1420) (792,0.1364) (945,0.1354) (1122,0.1319) (1362,0.1424) (1598,0.1245)};
\end{axis}
\end{tikzpicture}
  \end{center}
  \caption{The left plot shows the relation between the pruned and unpruned proof sizes.
    The right plot shows the time taken for the pruning operation against the
    total proof generation time.}
  \label{fig:proofprune}
  \end{figure}

	\section{Conclusions and further work}

  We have introduced a proof system for certifying the membership of a
  polynomial in a given multivariate polynomial ideal and have shown its
  correctness. For the problem of certifying the generating set of the
  automorphism group of a graph, we have presented a reduction to finding the
  Nullstellensatz certificate of a multivariate polynomial system over a finite
  field. Finally, we have come up with an explicit construction of such a
  certificate within the proof system. Based on this construction, we have
  developed a prototype implementation of a modification of the McKay-Piperno
  algorithm for constructing the proof, as well as an independent
  implementation of a proof verification procedure.

  An advantage of the approach outlined in the paper is that the proof
  verification procedure relies on straightforward polynomial operations and
  is decoupled from the specifics of the algorithm used to find the
  automorphisms. The downsides lie in the specific construction of the
  certificate which leads to very lengthy proofs for even moderately sized
  graphs. An obvious direction for further work is to come up with more
  efficient constructions. A path which seems likely to yield results is to
  find ways to eliminate the need to construct all of the polynomials
  referenced by Lemma \ref{lemma:monorefine}. Another reasonable direction
  would be to come up with constructions which work over the field
  $\mathbb{F}_2$, as this could potentially lead to significant performance
  gains. An additional interesting possibility is to explore combining an
  explicit proof construction with a general-purpose algorithm for finding
  Nullstellensatz certificates such as NulLA \cite{Nulla}, where some steps of
  the proof would be calculated using the general-purpose algorithm.

  Developing a formalized implementation of the proof verifier within an
  interactive theorem prover is another possible direction for further work.
  This would serve to eliminate any doubts about the correctness of the proof
  verifier implementation, and would further allow the automorphism group
  generation and verification to be used as primitives within formal proofs.

  It might be worth considering a prover implementation embedded within an
  existing solver down the line. This is, at present, not as practical as some
  of the other directions, considering that the primary bottleneck seems to be
  the proof generation time due to proof size. However, it could be a worthwhile
  endeavor in the future, assuming that more efficient proof constructions are
  discovered.

  Lastly, the status of the graph isomorphism problem within the complexity
  class $\textsc{coNP}$ remains an open question \cite{Babai2018}. Providing
  insight into the theoretical lower bounds for the sizes of the certificates
  could inform the search for more efficient proof constructions and their
  limitations.

  \section*{Acknowledgements}

  The author is grateful to Vesna Marinković for her guidance, suggestions, and
  repeated reviews of the manuscript throughout its preparation, and to Filip
  Marić and Milan Banković for reviewing early versions of the paper.

  The author is partially supported by the Ministry of  Science, Technological
  Development and Innovation, Republic of Serbia, through the project
  451-03-33/2026-03/200104.

  \section*{Declaration of generative AI use}

  During the preparation of this work the author used Claude (Anthropic), and
  ChatGPT and Codex (OpenAI) for tooling (i.e. helper Python scripts),
  mechanical edits to C++ code and manuscript source files, literature search,
  and reviews of the paper and the code. The author has closely reviewed and
  independently checked all AI-generated refactors and code, as well as all
  AI-suggested claims and references. The research, manuscript and
  implementation are the author's own and the author takes full responsibility
  for the content of the article and its accompanying artifacts.

	\bibliographystyle{alpha}
	\bibliography{references}

	\appendix

	\section{Proof of Lemma \ref{lemma:monorefine}}
  \label{app:proof}

	\begin{applem}[\textbf{\ref{lemma:monorefine}}]
    Let $\nu$ be a node of the search tree with partition mapping $\mu$.
    Polynomials of the form $(-1 + \sum_{v \in \mu[u]} x_{u, v}) M_\nu$ may be
    derived from $\mathcal{A}(G, a_1, \dots, a_m)$ for each vertex $u$, $(-1 +
    \sum_{u \in \mu^{-1}[v]} x_{u, v}) M_\nu$ may be derived for each vertex
    $v$, and $x_{u, v} M_\nu$ may be derived for each $u$ and $v$ such that $v
    \notin \mu[u]$.
	\end{applem}
	\begin{proof}
    Consider the partition mapping $\mu$ of node $\nu$. Starting from the root
    of the search tree, we may interpret the steps taken to calculate $\mu$ as
    a chain of individualization and refinement steps. Specifically, we
    consider the chain $\mu_0 \overset{i}{\mapsto} \mu_1
    \overset{s}{\mapsto} \dots \overset{s}{\mapsto} \mu_j
    \overset{i}{\mapsto} \mu_{j+1} \overset{s}{\mapsto} \dots
    \overset{s}{\mapsto} \mu$. We will show, using induction on the
    number of refinement and individualization steps, that the appropriate
    polynomials are derivable for each of the partition mappings in the given
    chain. We will use the notation $p = g_1 \underline{p_1} + \dots + g_m
    \underline{p_m}$ to denote that $p_1, \dots, p_m$ are assumed to be
    derivable and that the Ideal rule may be applied with coefficients $g_1,
    \dots, g_m$ to derive $p$.

    \begin{induction}

    \item[Base case] At the root of the search tree $\nu = []$ the
      bijective partition mapping maps the entire set of vertices to itself and no
      refinement or individualization steps have been performed. The
      $\operatorname{RowS}$ and $\operatorname{ColS}$ polynomials given in the
      system $\mathcal{A}(G, a_1, \dots, a_m)$ are exactly the polynomials from
      the statement and they are derivable by the Tautology rule.

    \item[Induction hypothesis] The polynomials $m_{i, u} = (-1 + \sum_{v \in
      \mu_i[u]} x_{u, v}) M_{\nu}$ and $w_{i, v} = (-1 + \sum_{u \in
      \mu_i^{-1}[v]} x_{u, v}) M_{\nu}$ are derivable for each vertex $u$ and
      $v$ in the bijective partition mapping $\mu_i:\mathbb{U}\to\mathbb{V}$ of node
      $\nu$, and the polynomials $e_{i, u, v} = x_{u, v} M_{\nu}$ are derivable
      for each pair of vertices $u$ and $v$ such that $v \notin \mu_i[u]$.

    \item[Induction step] We consider two different induction steps, one for
      individualization and one for refinement.

      (1) Assume $\mu_{i+1}$ is the partition mapping of $\nu' = [\nu, u_k]$ produced by
      individualizing vertex $u_k$ to vertex $v_k$ from $\mu_i$. Let $u_k \in U \in
      \mathbb{U}$ and $v_k \in V \in \mathbb{V}$. There are three cases to
      consider for deriving each of the three types of polynomials (Figure \ref{img:individualize}).
      \begin{figure}[H]
      \begin{center}
        \makeatletter
\newcommand{\myboxed}[1]{\mathpalette\MyBoxAux{#1}}
\newcommand{\MyBoxAux}[2]{%
  \tikz[baseline=(X.base)] 
    \node[draw, inner sep=1pt, anchor=base, text height=1.5ex, text depth=.25ex] 
    (X) {$\m@th#1#2$};%
}
\makeatother

\begin{tikzpicture}
  \tikzset{every label/.style={label distance=-1mm}}
  \begin{scope}[transform shape, scale=0.8]
    \node (polym) at (5, 1.2)[] {$m_{i+1}$};
    \node (a) at (0, 0)[label=90:(a)] {$\begin{smallmatrix}
      \dots & \{\ \myboxed{u}\ \} & \dots & \{u_k\} & \{\ \ \ \} & \dots & \{\ \ \ \} & \dots \\
      \dots & \myboxed{\{\ \ \ \}} & \dots & \{v_k\} & \{\ \ \ \} & \dots & \{\ \ \ \} & \dots
      \end{smallmatrix}$};
    \node (b) at (5, 0)[label=90:(b)] {$\begin{smallmatrix}
      \dots & \{\ \ \ \} & \dots & \{\myboxed{u_k}\} & \{\ \ \ \} & \dots & \{\ \ \ \} & \dots \\
      \dots & \{\ \ \ \} & \dots & \myboxed{\{v_k\}} & \{\ \ \ \} & \dots & \{\ \ \ \} & \dots
      \end{smallmatrix}$};
    \node (c) at (10, 0)[label=90:(c)] {$\begin{smallmatrix}
      \dots & \{\ \ \ \} & \dots & \{u_k\} & \{\ \myboxed{u} \ \} & \dots & \{\ \ \ \} & \dots \\
      \dots & \{\ \ \ \} & \dots & \{v_k\} & \myboxed{\{\ \ \ \}} & \dots & \{\ \ \ \} & \dots
      \end{smallmatrix}$};
    \node (polyw) at (5, -0.8)[] {$w_{i+1}$};
    \node (d) at (0, -2)[label=90:(d)] {$\begin{smallmatrix}
      \dots & \myboxed{\{\ \ \ \}}& \dots & \{u_k\} & \{\ \ \ \} & \dots & \{\ \ \ \} & \dots \\
      \dots & \{\ \myboxed{v} \ \} & \dots & \{v_k\} & \{\ \ \ \} & \dots & \{\ \ \ \} & \dots
      \end{smallmatrix}$};
    \node (e) at (5, -2)[label=90:(e)] {$\begin{smallmatrix}
      \dots & \{\ \ \ \} & \dots & \myboxed{\{u_k\}} & \{\ \ \ \} & \dots & \{\ \ \ \} & \dots \\
      \dots & \{\ \ \ \} & \dots & \{\myboxed{v_k}\} & \{\ \ \ \} & \dots & \{\ \ \ \} & \dots
      \end{smallmatrix}$};
    \node (f) at (10, -2)[label=90:(f)] {$\begin{smallmatrix}
      \dots & \{\ \ \ \} & \dots & \{u_k\} & \myboxed{\{\ \ \ \}} & \dots & \{\ \ \ \} & \dots \\
      \dots & \{\ \ \ \} & \dots & \{v_k\} & \{\ \myboxed{v} \ \} & \dots & \{\ \ \ \} & \dots
      \end{smallmatrix}$};
    \node (polye) at (5, -2.8)[] {$e_{i+1}$};
    \node (g) at (0, -4)[label=90:(g)] {$\begin{smallmatrix}
      \dots & \{\ \myboxed{u} \ \} & \dots & \{u_k\} & \{\ \ \ \} & \dots & \{\ \ \ \} & \dots \\
      \dots & \{\ \ \ \} & \dots & \{v_k\} & \{\ \ \ \} & \dots & \{\ \myboxed{v} \ \} & \dots
      \end{smallmatrix}$};
    \node (h) at (5, -4)[label=90:(h)] {$\begin{smallmatrix}
      \dots & \{\ \ \ \} & \dots & \{\myboxed{u_k}\} & \{\ \ \ \} & \dots & \{\ \ \ \} & \dots \\
      \dots & \{\ \ \ \} & \dots & \{v_k\} & \{\ \myboxed{v} \ \} & \dots & \{\ \ \ \} & \dots
      \end{smallmatrix}$};
    \node (i) at (10, -4)[label=90:(i)] {$\begin{smallmatrix}
      \dots & \{\ \ \ \} & \dots & \{u_k\} & \{\ \myboxed{u} \ \} & \dots & \{\ \ \ \} & \dots \\
      \dots & \{\ \ \ \} & \dots & \{\myboxed{v_k}\} & \{\ \ \ \} & \dots & \{\ \ \ \} & \dots
      \end{smallmatrix}$};
  \end{scope}
\end{tikzpicture}
      \end{center}
      \caption{The vertex $u_k$ is individualized to $v_k$. The partition
        mapping $\mu_i$ is shown, with a different combination of vertices and
        cells being highlighted in each one.}
      \label{img:individualize}
      \end{figure}
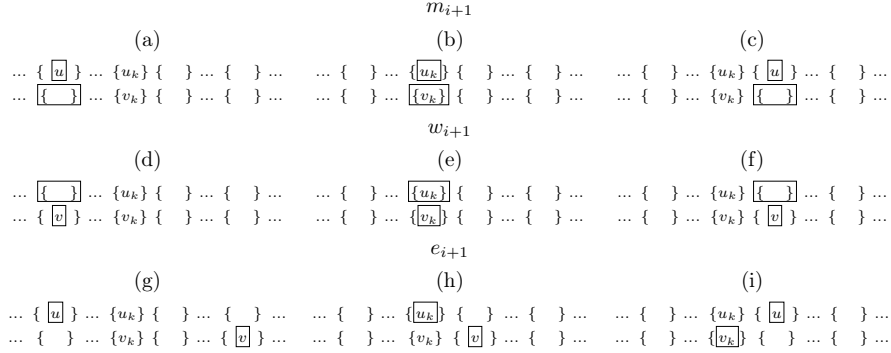

      The polynomials $m_{i+1, u}$ are derivable for each vertex $u$. Let
      $S^m_{i, u} = -1 + \sum_{v \in \mu_i[u]} x_{u, v}$.

        (a) If $u \notin U$ then
        \begin{align*}m_{i+1, u} &= S^m_{i+1, u} M_{\nu'} \\
          &= S^m_{i, u} M_{\nu'} \\
          &= S^m_{i, u} M_{\nu} x_{u_k, v_k} \\
          &= x_{u_k, v_k} \underline{m_{i, u}}
        \end{align*}

        (b) If $u = u_k$ then
        \begin{align*}m_{i+1, u} &= S^m_{i+1, u} M_{\nu'} \\
          &= (-1 + x_{u_k, v_k}) M_{\nu'} \\
          &= (-1 + x_{u_k, v_k}) M_{\nu} x_{u_k, v_k} \\
          &= M_{\nu} \underline{\operatorname{Bin}_{u_k, v_k}}
        \end{align*}

        (c) If $u \in U \setminus \{u_k\}$ then
        \begin{align*}m_{i+1, u} &= S^m_{i+1,u} M_{\nu'} \\
          &= (S^m_{i,u}-x_{u,v_k}) M_{\nu'} \\
          &= S^m_{i,u} M_{\nu'}-x_{u,v_k} M_{\nu'} \\
          &= S^m_{i,u} M_{\nu} x_{u_k, v_k} - e_{i+1, u, v_k} \\
          &= x_{u_k, v_k} \underline{m_{i, u}} - \underline{e_{i+1, u, v_k}}
        \end{align*}

      The polynomials $w_{i+1, v}$ are derivable for each vertex $v$. Let
      $S^w_{i, v} = -1 + \sum_{u \in \mu_i^{-1}[v]} x_{u, v}$.

        (d) If $v \notin V$ then
        \begin{align*}w_{i+1, v} &= S^w_{i+1, v} M_{\nu'} \\
          &= S^w_{i, v} M_{\nu'} \\
          &= S^w_{i, v} M_{\nu} x_{u_k, v_k} \\
          &= x_{u_k, v_k} \underline{w_{i, v}}
        \end{align*}

        (e) If $v = v_k$ then
        \begin{align*}w_{i+1, v} &= S^w_{i+1, v} M_{\nu'} \\
          &= (-1 + x_{u_k, v_k}) M_{\nu'} \\
          &= \underline{m_{i+1, u_k}}
        \end{align*}

        (f) If $v \in V \setminus \{v_k\}$ then
        \begin{align*}w_{i+1, v} &= S^w_{i+1,v} M_{\nu'} \\
          &= (S^w_{i,v}-x_{u_k,v}) M_{\nu'} \\
          &= S^w_{i,v} M_{\nu'}-x_{u_k,v} M_{\nu'} \\
          &= S^w_{i,v} M_{\nu} x_{u_k, v_k} - e_{i+1, u_k, v} \\
          &= x_{u_k, v_k} \underline{w_{i, v}} - \underline{e_{i+1, u_k, v}}
        \end{align*}

      The polynomials $e_{i+1, u, v}$ are derivable for each pair of vertices
      $u$ and $v$ such that $v \notin \mu_{i+1}[u]$.

        (g) If $u \notin U$ or $v \notin V$ then
        \begin{align*}e_{i+1, u, v} &= x_{u, v} M_{\nu'} \\
          &= x_{u, v} M_{\nu} x_{u_k, v_k} \\
          &= x_{u_k, v_k} \underline{e_{i, u, v}}
        \end{align*}

        (h) If $u = u_k$ and $v \in V \setminus \{v_k\}$ then
        \begin{align*}e_{i+1, u, v} &= x_{u_k, v} M_{\nu'} \\
          &= x_{u_k, v} M_{\nu} x_{u_k, v_k} \\
          &= M_{\nu} \underline{\operatorname{RowNZ}_{u_k; v, v_k}}
        \end{align*}

        (i) If $v = v_k$ and $u \in U \setminus \{u_k\}$ then
        \begin{align*}e_{i+1, u, v} &= x_{u, v_k} M_{\nu'} \\
          &= x_{u, v_k} M_{\nu} x_{u_k, v_k} \\
          &= M_{\nu} \underline{\operatorname{ColNZ}_{u, u_k; v_k}}
        \end{align*}

    (2) Assume $\mu_{i+1}$ is produced by splitting $\mu_i$ according to some
    cell $W \in \mathbb{U}$. We will prove a series of auxiliary statements
    that will help us derive the required polynomials for the partition mapping
    $\mu_{i+1}$.

    Let $I_{u, W} = \sum_{v \in \mu_{i}(W)} x_{u, v}$ and $S_{u, W} = I_{u, W}
    - 1$. If $u \in W$, then $S_{u, W} M_{\nu}$ is derivable since $S_{u, W}
    M_{\nu} = \underline{m_{i, u}}$. If $u \notin W$, then $I_{u, W} M_{\nu}$ is derivable
    since $I_{u, W} M_{\nu} = \sum_{v \in \mu_i(W)} x_{u, v} M_{\nu} = \sum_{v
    \in \mu_i(W)} \underline{e_{i, u, v}}$.

    We will now show two equalities related to the sum $\sum_{v \in \mu_i(W)} \operatorname{Iso}_{u, v}$.
    For any vertex $u$ we have:

    \begin{minipage}[t]{0.45\textwidth}
    \begin{align*}
      \sum_{\mathclap{v \in \mu_i(W)}} (AX)_{u, v} &= \smashoperator[l]{\sum_{v \in \mu_i(W)}} \sum_{t = 0}^{n-1} a_{u, t} x_{t, v} \\
      &= \sum_{t=0}^{n-1} a_{u, t} \sum_{\mathclap{v \in \mu_i(W)}} x_{t, v} \\
      &= \sum_{t=0}^{n-1} a_{u, t} I_{t, W} \\
      &= \sum_{\mathclap{t \in n_G(u)}} I_{t, W}
    \end{align*}
    \end{minipage}
    \hspace{0.05\textwidth}
    \begin{minipage}[t]{0.45\textwidth}
    \begin{align*}
      \sum_{\mathclap{v \in \mu_i(W)}} (XA)_{u, v} &= \smashoperator[l]{\sum_{v \in \mu_i(W)}} \sum_{t=0}^{n-1} x_{u, t} a_{t, v} \\
      &= \sum_{t=0}^{n-1} x_{u, t} \sum_{\mathclap{v \in \mu_i(W)}} a_{t, v} \\
      &= \sum_{t=0}^{n-1} x_{u, t} \operatorname{deg}(t, \mu_i(W))
    \end{align*}
    \end{minipage}

    Let $d_{u, W} = \operatorname{deg}(u, W) - \sum_{v \in \mu_i[u]} x_{u, v}
    \operatorname{deg}(v, \mu_i(W))$. For any vertex $u$ we show that $d_{u, W}
    M_{\nu}$ is derivable:

    \begin{align*}
       & M_{\nu}\sum_{\mathclap{v \in \mu_i(W)}} \underline{\operatorname{Iso}_{u, v}} - \sum_{\mathclap{t \in n_G(u) \setminus W}} \underline{I_{t, W} M_{\nu}} - \sum_{\mathclap{t \in n_G(u) \cap W}} \underline{S_{t, W} M_{\nu}} + \sum_{\mathclap{v \notin \mu_i[u]}} \operatorname{deg}(v, \mu_i(W)) \underline{e_{i, u, v}}\\
      =& M_{\nu}\sum_{\mathclap{v \in \mu_i(W)}} \operatorname{Iso}_{u, v} - \sum_{\mathclap{t \in n_G(u) \setminus W}} I_{t, W} M_{\nu} - \sum_{\mathclap{t \in n_G(u) \cap W}} S_{t, W} M_{\nu} + \sum_{\mathclap{v \notin \mu_i[u]}} \operatorname{deg}(v, \mu_i(W)) M_\nu x_{u, v}\\
      =& M_{\nu}\sum_{\mathclap{v \in \mu_i(W)}} \operatorname{Iso}_{u, v} - \sum_{\mathclap{t \in n_G(u)}} I_{t, W} M_\nu + \sum_{\mathclap{t \in n_G(u) \cap W}} M_\nu + \sum_{\mathclap{v \notin \mu_i[u]}} \operatorname{deg}(v, \mu_i(W)) M_\nu x_{u, v} \\
      =& M_{\nu}\sum_{\mathclap{v \in \mu_i(W)}} (AX-XA)_{u, v} - \sum_{\mathclap{t \in n_G(u)}} I_{t, W} M_\nu + \sum_{\mathclap{t \in n_G(u) \cap W}} M_\nu + \sum_{\mathclap{v \notin \mu_i[u]}} \operatorname{deg}(v, \mu_i(W)) M_\nu x_{u, v} \\
      =& M_{\nu}\sum_{\mathclap{v \in \mu_i(W)}} -(XA)_{u, v} + \sum_{\mathclap{t \in n_G(u) \cap W}} M_\nu + \sum_{\mathclap{v \notin \mu_i[u]}} \operatorname{deg}(v, \mu_i(W)) M_\nu x_{u, v} \\
      =& -M_{\nu}\sum_{\mathclap{t=0}}^{n-1} x_{u, t} \operatorname{deg}(t, \mu_i(W)) + \sum_{\mathclap{t \in n_G(u) \cap W}} M_\nu + \sum_{\mathclap{v \notin \mu_i[u]}} \operatorname{deg}(v, \mu_i(W)) M_\nu x_{u, v} \\
      =& \sum_{\mathclap{t \in n_G(u) \cap W}} M_\nu - M_{\nu}\sum_{\mathclap{v \in \mu_i[u]}} x_{u, v} \operatorname{deg}(v, \mu_i(W)) \\
      =& M_\nu \operatorname{deg}(u, W) - M_{\nu}\sum_{\mathclap{v \in \mu_i[u]}} x_{u, v} \operatorname{deg}(v, \mu_i(W)) \\
      =& M_\nu d_{u, W}
    \end{align*}

    We can now show how the three types of polynomials may be derived for
    $\mu_{i+1}$. If $u$ and $v$ are vertices such that $v \notin \mu_{i}[u]$,
    then trivially $e_{i+1, u, v} = \underline{e_{i, u, v}}$. If, on the other
    hand, $v \in \mu_i[u]$ but $v \notin \mu_{i+1}[u]$, we show that $e_{i+1,
    u, v}$ is derivable. Let $c = \operatorname{deg}(u, W) -
    \operatorname{deg}(v, \mu_i(W))$. Note that $c$ is invertible in
    $\mathbb{F}_p$ since $0 \leq \operatorname{deg}(u, W),
    \operatorname{deg}(v, \mu_i(W)) < n < p$ and since $v \in \mu_i[u]
    \setminus \mu_{i+1}[u]$ implies that $\operatorname{deg}(u, W) \neq
    \operatorname{deg}(v, \mu_i(W))$. We have:

    \begin{align*}
       & (x_{u, v} \underline{d_{u, W} M_\nu} + M_\nu \sum_{\mathclap{v' \in \mu_i[u] \setminus \{v\}}} \operatorname{deg}(v', \mu_i(W)) \underline{\operatorname{RowNZ}_{u; v, v'}} + \operatorname{deg}(v, \mu_i(W)) M_\nu \underline{\operatorname{Bin}_{u, v}}) \\
      =& (d_{u, W} + \sum_{\mathclap{v' \in \mu_i[u] \setminus \{v\}}} \operatorname{deg}(v', \mu_i(W)) x_{u, v'} + \operatorname{deg}(v, \mu_i(W)) (x_{u, v} - 1)) x_{u, v} M_\nu \\
      =& (\operatorname{deg}(u, W) - \sum_{\mathclap{v' \in \mu_i[u]}} x_{u, v'} \operatorname{deg}(v', \mu_i(W)) + \sum_{\mathclap{v' \in \mu_i[u] \setminus \{v\}}} \operatorname{deg}(v', \mu_i(W)) x_{u, v'} \\
       & \qquad + \operatorname{deg}(v, \mu_i(W)) (x_{u, v} - 1)) x_{u, v} M_\nu \\
      =& (\operatorname{deg}(u, W) - x_{u, v} \operatorname{deg}(v, \mu_i(W)) + \operatorname{deg}(v, \mu_i(W)) (x_{u, v} - 1)) x_{u, v} M_\nu \\
      =& (\operatorname{deg}(u, W) - \operatorname{deg}(v, \mu_i(W))) x_{u, v} M_\nu \\
      =& c e_{i+1, u, v}
    \end{align*}
    From this, we can easily derive $e_{i+1, u, v}$ by multiplying with
    $c^{-1}$.

    We may derive $m_{i+1, u}$ for any vertex $u$ and $w_{i+1, v}$ for any vertex $v$ as:

    \begin{minipage}[t]{0.45\textwidth}
    \begin{align*}
      m_{i+1, u} &= S^m_{i+1, u} M_\nu \\
      &= S^m_{i, u} M_\nu - \sum_{\mathclap{v \in \mu_i[u] \setminus \mu_{i+1}[u]}} x_{u, v} M_\nu \\
      &= \underline{m_{i, u}} - \sum_{\mathclap{v \in \mu_i[u] \setminus \mu_{i+1}[u]}} \underline{e_{i+1, u, v}}
    \end{align*}
    \end{minipage}
    \begin{minipage}[t]{0.45\textwidth}
    \begin{align*}
      w_{i+1, v} &= S^w_{i+1, v} M_\nu \\
      &= S^w_{i, v} M_\nu - \sum_{\mathclap{u \in \mu^{-1}_i[v] \setminus \mu^{-1}_{i+1}[v]}} x_{u, v} M_\nu \\
      &= \underline{w_{i, v}} - \sum_{\mathclap{u \in \mu^{-1}_i[v] \setminus \mu^{-1}_{i+1}[v]}} \underline{e_{i+1, u, v}}
    \end{align*}
    \end{minipage}

    \end{induction}
			
	\end{proof}

\end{document}